\documentclass{ieeeaccess}
\usepackage{cite}
\usepackage{amsmath,amssymb,amsfonts}
\usepackage{amsthm}
\usepackage{mathtools}
\usepackage{commath}

\usepackage{pifont}
\newcommand{\cmark}{\ding{51}}%
\newcommand{\xmark}{\ding{55}}%

\usepackage{graphicx}
\usepackage[font={sf,scriptsize},           
            labelfont={bf,color=accessblue},
            caption=false]                  
           {subfig}

\usepackage{textcomp}
\usepackage{multirow}

\usepackage{amsthm}

\usepackage{algorithm}
\usepackage{algpseudocode}

\usepackage{array}

\makeatletter
\renewenvironment{proof}[1][\proofname]{\par
  \normalfont
  \topsep3\p@\@plus3\p@ \trivlist
  \item[\hskip3\labelsep\itshape
    #1\@addpunct{.}]\ignorespaces
}{%
  \qed\endtrivlist
}
\makeatother

\newtheoremstyle{note}
{3pt}
{3pt}
{}
{\parindent}
{\itshape}
{:}
{.5em}
{}

\theoremstyle{note}
\newtheorem{lem}{Lemma}
\newtheorem{rem}{Remark}

\DeclareMathOperator*{\argmax}{arg\,max}

\def\BibTeX{{\rm B\kern-.05em{\sc i\kern-.025em b}\kern-.08em
    T\kern-.1667em\lower.7ex\hbox{E}\kern-.125emX}}
    
\expandafter\def\expandafter\normalsize\expandafter{%
    \normalsize
    \setlength\abovedisplayskip{-4pt}
    \setlength\belowdisplayskip{2pt}
    \setlength\abovedisplayshortskip{-3pt}
    \setlength\belowdisplayshortskip{2pt}
}

\algnewcommand\algorithmicinput{\textbf{Output:}}
\algnewcommand\Output{\item[\algorithmicinput]}

\begin{document}

\renewcommand{\algorithmicrequire}{\textbf{Input:}}
\renewcommand{\algorithmicensure}{\textbf{Initialize:}}
\let\oldReturn\Return
\renewcommand{\Return}{\State\oldReturn}

\history{Date of publication xxxx 00, 0000, date of current version xxxx 00, 0000.}
\doi{10.1109/ACCESS.2017.DOI}

\title{An Efficient Interference-Aware Constrained Massive MIMO Beamforming for mm-Wave JSDM}
\author{\uppercase{Murat Bayraktar}, \IEEEmembership{Student Member, IEEE},
\uppercase{Gokhan M. Guvensen}, \IEEEmembership{Member, IEEE}}
\address{Middle East Technical University, Ankara, Turkey (e-mail: muratbay@metu.edu.tr, guvensen@metu.edu.tr)}

\tfootnote{``This work was supported in part by TUBITAK under Project 218E039''}

\markboth
{Bayraktar \headeretal: An Efficient Interference-Aware Constrained Massive MIMO Beamforming for mm-Wave JSDM}
{Bayraktar \headeretal: An Efficient Interference-Aware Constrained Massive MIMO Beamforming for mm-Wave JSDM}

\corresp{Corresponding author: Murat Bayraktar (e-mail: muratbay@metu.edu.tr).}

\begin{abstract}

Low-complexity beamformer design with practical constraints is an attractive research area for hybrid analog/digital systems in mm-wave massive multiple-input multiple-output (MIMO). This paper investigates interference-aware pre-beamformer (analog beamformer) design for joint spatial division and multiplexing (JSDM) which is a user-grouping based two-stage beamforming method. Single-carrier frequency domain equalization (SC-FDE) is employed in uplink frequency-selective channels. First, unconstrained slowly changing statistical analog beamformer of each group, namely, generalized eigenbeamformer (GEB) which has strong interference suppression capability is designed where the mutual information in reduced dimension is maximized. Then, constant-modulus constrained approximations of unconstrained beamformer are obtained by utilizing alternating minimization algorithms for fully connected arrays and fixed subarrays. In addition, a dynamic subarray algorithm is proposed where the connections between radio frequency (RF) chains and antennas are changed with changing channel statistics. Convergence of the proposed alternating minimization-based algorithms are provided along with their complexity analysis. It is observed that additional complexity of proposed algorithms is insignificant for the overall system design. Although most of the interference is suppressed with the help of proposed constrained beamformers, there may be some residual interference after the analog beamforming stage. Therefore, linear minimum mean square error (LMMSE) type digital beamformers, which take the residual interference in reduced dimension into account, are proposed instead of zero-forcing (ZF) type. Simulation results verify the superiority of the proposed interference-aware constrained design over existing approaches in terms of beampattern, spectral efficiency, outage capacity and channel estimation accuracy.

\end{abstract}


\titlepgskip=-16pt

\maketitle

\section{Introduction}
\label{sec:introduction}

\PARstart{M}{assive} multiple-input multiple-output (MIMO) at mm-wave bands is the frontier technology for next generation wireless communication systems due to the availability of vast spectrum \cite{mMIMO,mMIMOoverview}. Hybrid analog/digital structure is an efficient approach for these systems as it has lower hardware complexity and energy consumption compared to fully digital beamforming and mm-wave bands exhibit joint angle-delay sparsity \cite{SPmmWave}. Hence, researchers have been working on low-complexity beamformer design for hybrid systems. \cite{SpaSparse,LowComplexityHybrid,AltMin,WeiYu,HybridSurvey}. Joint spatial division and multiplexing (JSDM) is a user-grouping based two-stage beamforming method where analog beamformers are designed by using slowly varying channel statistics, i.e., channel covariance matrices (CCMs) \cite{JSDM,JSDMmmWave}. User-grouping reduces signal processing complexity since digital processing of each group becomes independent from each other. In addition, only instantaneous intra-group channels in reduced dimension need to be learned which reduces channel estimation overhead. In this paper, we adopt JSDM framework in frequency-selective channels where we employ single-carrier frequency domain equalization (SC-FDE) for uplink transmission.

\subsection{Related Work}

Most of the existing researches for massive MIMO systems adopt flat-fading channel model by considering the use of orthogonal frequency division multiplexing (OFDM) \cite{mMIMOoverview}. The reason to prefer single-carrier (SC) modulation over OFDM in this paper is its superiority for systems having nonlinearities, such as quantized MIMO \cite{MessagePassing,MUD}, owing to its lower peak-to-average power ratio (PAPR) \cite{PAPR} and robustness to carrier-frequency-offset (CFO) errors \cite{SC-FDE}. Having lower PAPR is a critical property for systems with non-linear elements \cite{SCvsOFDM}, where the advantages of SC transmission over the multicarrier are demonstrated. Moreover, even for unquantized linear systems, there are many recent studies that motivate the use of SC especially for MIMO systems in mm-wave bands \cite{SCoptimality,mmWaveSC_Hybrid_mMIMO,GC,Anil}. In these studies, the mitigation of inter-symbol interference (ISI) is fulfilled via reduced complexity signal processing.

Analog beamformer design for JSDM framework is a critical task since inter-group interference should be mitigated at this stage. There are several papers that find unconstrained analog beamformers that take interference into account \cite{JSDM,JSDMmmWave,NewJSDM}. However, they do not suppress inter-group interference perfectly. In \cite{NewJSDM}, a weighted minimum mean square error (WMMSE) type digital beamformer is proposed to suppress residual interference. A near-optimal generalized eigenbeamformer (GEB), which has strong interference suppression capability, is proposed for SC transmission at mm-wave bands in \cite{GC,Anil}. However, none of these analog beamformers obey constant-modulus constraint. For two-stage beamforming concept, DFT beamformer and phase of unconstrained beamformers are considered as state-of-the-art constrained analog beamformer alternatives \cite{JSDM,DFT,TwoStage,Exploiting,Dynamic}.

On the other hand, constrained analog beamformer design for hybrid analog/digital systems is a well-studied research area \cite{LowComplexityHybrid,AltMin,WeiYu}. These papers use the instantaneous full dimensional channel state information (CSI) for beamformer design with fully connected arrays. In \cite{Exploiting}, two-stage beamforming approach is adopted where columns of the unconstrained analog beamformer are selected as eigenvectors of sum of CCMs of users. Furthermore, constrained analog beamformer is obtained via alternating minimization algorithm and a compensation matrix is applied in digital baseband to orthogonalize the overall analog beamformer. On the other hand, compared to fully connected arrays, partially connected ones require much less phase shifters at the expense of spectral efficiency. Compromise between spectral efficiency and energy consumption for fully connected arrays and fixed subarrays is extensively studied in \cite{AltMin,FixedPart}.

Recently, dynamic subarray design, which is a challenging problem, is studied where CCMs are used to obtain the optimal connection between radio frequency (RF) chains and antennas for partially connected arrays \cite{Dynamic,ChStat,DynamicWideband,DynamicWideband2,DynamicMISO}. In these papers, it is shown that dynamic subarrays outperform fixed ones in terms of several metrics such as spectral efficiency.\cite{Dynamic,ChStat,DynamicWideband,DynamicWideband2} consider a point-to-point hybrid analog/digital system whereas \cite{DynamicMISO} considers multiple users. For dynamic subarrays, connection change can be realized with switches which are shown to be energy-efficient components for mm-wave massive MIMO systems \cite{Switch1,Switch2}.

Table \ref{RelatedWork} summarizes the contributions of the existing work on constrained beamformer design. None of them are interested in constrained pre-beamformer (analog beamformer) design even for fully connected arrays where mitigation of interference is realized at this stage. That is, interference nulling capability of statistical pre-beamformers is ignored while considering constant-modulus constraint in prior work. Furthermore, dynamic subarray design is not considered for interference-aware pre-beamformers in the literature.

\begin{table*}[!h]
\small
\caption{Summary of existing work on constrained beamformer design with fully connected and partially connected arrays (fixed and dynamic subarrays)}
\vspace{-5mm}
\begin{center}\begin{tabular}{| >{\centering}m{28mm}|c|c|c|c|c|c|c|c|c|c|c|c|}
\hline
Papers & \cite{LowComplexityHybrid} & \cite{AltMin} & \cite{WeiYu} & \cite{JSDM,JSDMmmWave,NewJSDM} & \cite{GC,Anil} & \cite{Exploiting} & \cite{Dynamic} & \cite{FixedPart} & \cite{ChStat} & \cite{DynamicWideband,DynamicWideband2} & \cite{DynamicMISO} & This Work \\ \hline
Frequency-Selectivity & \xmark & \cmark & \xmark & \xmark & \cmark & \xmark & \cmark & \xmark & \xmark & \cmark & \xmark & \cmark \\ \hline
Multiuser & \xmark & \xmark & \cmark & \cmark & \cmark & \cmark & \xmark & \xmark & \xmark & \xmark & \cmark & \cmark \\ \hline
Two-Stage Beamforming (Spatial Statistical Pre-Beamformer) & \xmark & \xmark & \cmark & \cmark & \cmark & \cmark & \cmark & \cmark & \xmark & \xmark & \cmark & \cmark \\ \hline
Interference-Aware Pre-Beamformer & \xmark & \xmark & \xmark & \cmark & \cmark & \xmark & \xmark & \cmark & \xmark & \xmark & \xmark & \cmark \\ \hline
Constrained Analog Beamformer & \cmark & \cmark & \cmark &\xmark & \xmark & \cmark & \xmark & \cmark & \cmark & \cmark & \cmark & \cmark \\ \hline
Fully Connected Array & \cmark & \cmark & \cmark & \cmark & \cmark & \cmark & \cmark & \xmark & \xmark & \xmark & \xmark & \cmark \\ \hline
Fixed Subarray & \xmark & \cmark & \xmark & \xmark & \xmark & \xmark & \cmark & \cmark & \cmark & \cmark & \cmark & \cmark \\ \hline
Dynamic Subarray & \xmark & \xmark & \xmark & \xmark & \xmark & \xmark & \cmark & \xmark & \cmark & \cmark & \cmark & \cmark \\ \hline
\end{tabular}
\vspace{-5mm}
\label{RelatedWork}
\end{center}
\end{table*}

\subsection{Contributions}

In this paper, we design interference-aware analog beamformers for JSDM framework utilizing SC-FDE in uplink transmission. Our contributions in this paper are as follows:

\begin{itemize}

\item We firstly design a near-optimal unconstrained analog beamformer that considers interference. Unlike most prior work, we concentrate on analog beamformers which are updated with slowly varying CCMs. We try to maximize mutual information between the frequency domain received signal and intended group's signal in reduced dimension, i.e., after the analog beamformer. Although we employ a frequency-selective channel model, cost function turns out to be independent of frequency bin for spatial-narrowband arrays. Near-optimal unconstrained analog beamformer that maximizes the mutual information in reduced dimension coincides with GEB concept. This type of beamformers have strong interference capability and their optimality can be shown with respect to several criteria \cite{GC,Anil}.

\item We propose algorithms to obtain constant-modulus constrained approximations of GEB for both fully and partially connected arrays. We decompose the analog beamformer into two stages, namely, constrained analog beamformer and compensation matrix. Compensation matrix is applied in the digital baseband just before digital beamformers. Compensation matrix increases the degrees of freedom in the optimization problem which is formulated as minimization of the Euclidean distance between GEB and combination of constrained analog beamformer and compensation matrix. We obtain a constrained solution for fully connected arrays by utilizing an alternating minimization-based algorithm. Then, we present the algorithm to find the constrained solution for partially connected arrays with any fixed connection structure. Furthermore, we propose a low-complexity algorithm to find the optimal partially connected array structure that yields the maximum expected signal-to-interference-plus-noise ratio (SINR) in reduced dimension by exploiting the properties of GEB. The last problem is known as dynamic subarray design. To the authors' best knowledge this is the first dynamic subarray algorithm for slowly varying statistical analog beamformers in JSDM framework. In addition, we promoted the usage of linear minimum mean square error (LMMSE) type digital beamformers that consider the residual interference after constrained analog beamformers that cannot mitigate the inter-group interference completely. In this way, interference is taken into account in both analog and digital beamforming stages.

\item We provide a comprehensive analysis by using several performance measures; beampattern, spectral efficiency, outage capacity and accuracy of channel estimation in reduced dimension. Beampattern analysis of constrained analog beamformers, especially partially connected arrays, is lacking in the literature. Our analysis provides novel insights about the interference suppression capabilities of constrained analog beamformers.

\end{itemize}

In summary, we design interference-aware slowly varying analog beamformers in JSDM framework unlike most prior work as it can be seen from Table \ref{RelatedWork}. However, our design can be used in any framework with the presence of interference whose covariance matrix is known. In addition, proposed algorithms are evaluated with simulations by using the provided analysis tools. We consider a scenario where a group is assumed to be mobile so that its multipath components (MPCs) can be overlapped with angular regions of other groups which leads to increase in inter-group interference at certain situations. In this way, interference suppression capabilities of constrained analog beamformers can be observed. The results show that proposed fully connected array algorithm is superior to DFT beamformer and phase of GEB which are considered as state-of-the-art. Proposed algorithm attains the performance of GEB with moderate interference strength. Furthermore, proposed dynamic subarray algorithm outperforms commonly used fixed subarray structures.

\textit{Notations:} Scalars, column vectors and matrices are denoted by lowercase (e.g., $ x $), lower-case boldface (e.g., $ \textbf{x} $) and uppercase boldface (e.g., $ \textbf{X} $) letters, respectively. $ | x | $ and $ x^* $ are the magnitude and complex conjugate of scalar $ x $, respectively. $ \textbf{X}^T $, $ \textbf{X}^H $, $ \textbf{X}^{-1} $ and $ | \textbf{X} | $ represent the transpose, Hermitian, inverse and determinant of matrix $ \textbf{X} $, respectively. $ [ \textbf{X} ]_{(i,j)} $ is the entry of matrix $ \textbf{X} $ at $ i^{th} $ row and $ j^{th} $ column. $ [ \textbf{X} ]_{(i,:)} $ and $ [ \textbf{X} ]_{(:,j)} $ are used to extract $ i^{th} $ row and $ j^{th} $ column of matrix $ \textbf{X} $, respectively. $ \textbf{I}_{N} $ is the identity matrix with size $ N \times N $. $ \mathbb{E} \{ \cdot \} $ and $ \mathrm{Tr} \{ \cdot \} $ are the expectation and trace operators, respectively. $ \norm{\textbf{x}} $ and $ \norm{\textbf{X}}_F $ denote the Euclidean norm of vector $ \textbf{x} $ and Frobenius norm of matrix $ \textbf{X} $, respectively. $ \measuredangle( \cdot ) $ extracts the phase values of given input. $ \textrm{blkdiag} \{ \cdot \} $ is used to construct a block diagonal matrix with given inputs. $ \mathcal{CN}\big(\textbf{x},\textbf{R}\big) $ is a complex Gaussian random vector with mean $ \textbf{x} $ and covariance $ \textbf{R} $. $ P ( \cdot ) $ represents the probability value. $ \otimes $ denote the Kronecker product between two matrices. $ \delta_{ij} $ is the Kronecker delta function.

\section{System Model}
\label{sec:sysmod}

We consider an uplink SC-FDE system where the base station (BS) with $ M $ antennas serves $ K $ single-antenna users. Total number of RF chains is denoted by $ D $. Users are partitioned into $ G $ groups according to their channel covariance eigenspaces\footnote{The design of user grouping algorithms is out of scope of this paper. Efficient user-grouping procedures can be found in \cite{JSDM_Group,Joint_Group}.}. Number of users in group $ g $ is denoted by $ K_g $ and number of RF chains allocated to this group is $ D_g $. 

\subsubsection{Wideband SC-FDE Uplink Transmission}

Received signal vector $ \textbf{y}_n \in \mathbb{C}^{M \times 1} $ at the BS at $ n^{th} $ time instance is expressed as

\begin{equation}
\textbf{y}_n = \sum_{g=1}^{G} \sum_{l=0}^{L-1} \textbf{H}_{l}^{(g)} \textbf{x}_{(n-l)_N}^{(g)} + \textbf{n}_n, \label{y_n}
\end{equation}

\noindent
for $ n = 0,1,\dots,N-1 $, where block length is denoted by $ N $. In \eqref{y_n}, $ \textbf{H}_{l}^{(g)} \triangleq \big[ \textbf{h}_l^{(g_1)}, \dots, \textbf{h}_l^{(g_{K_g})} \big] \in \mathbb{C}^{M \times K_g} $ represent the channel matrix of group $ g $ at $ l^{th} $ MPC where the channel vector of user $ m $ in group $ g $ is denoted by $ \textbf{h}_l^{(g_m)} $ and the total number of MPCs is denoted by $ L $. A cyclic prefix with length larger than $ L $ is used to obtain circulant channel matrices while preventing the inter-block interference. Transmitted symbol vector of group $ g $ is defined as $ \textbf{x}_n^{(g)} \triangleq \big[ x_n^{(g_1)}, \dots, x_n^{(g_{K_g})} \big]^T \in \mathbb{C}^{K_g \times 1} $ with $ \mathbb{E} \big\{ x_n^{(g_m)} \big( x_{n'}^{(g'_{m'})} \big)^* \big\} = \frac{E_s^{(g)}}{K_g} \delta_{gg'} \delta_{mm'} \delta_{nn'} $ where $ x_n^{(g_m)} $ is the transmitted symbol of user $ m $ in group $ g $. Total transmit energy of group $ g $ is denoted by $ E_s^{(g)} $. Noise vector $ \textbf{n}_n $ is comprised of zero-mean circularly symmetric complex Gaussian random variables with variance $ N_0 $.

Channel vectors obey the spatially correlated Rayleigh channel model with $ \textbf{h}_l^{(g_m)} \sim \mathcal{CN}\big(\textbf{0},\textbf{R}_l^{(g_m)}\big) $ where $ \textbf{R}_l^{(g_m)} $ is the CCM of user $ m $ in group $ g $ at $ l^{th} $ MPC \cite{JSDM,JSDMmmWave,GC,Anil,NewJSDM,Exploiting}. User channel vectors are assumed to be mutually uncorrelated among MPCs and users which is expressed as

\begin{equation}
\mathbb{E}\bigg\{\textbf{h}_l^{(g_{m})} \Big[\textbf{h}_{l'}^{(g'_{m'})}\Big]^H \bigg\} = \textbf{R}_l^{(g_{m})} \delta_{gg'} \delta_{mm'} \delta_{ll'}. \label{E(h_l)}
\end{equation}

Considering the one-ring scattering model \cite{JSDM,JSDMmmWave,Exploiting}, CCM of a user at $ l^{th} $ MPC is defined as

\begin{equation}
\textbf{R}_l^{(g_m)} \triangleq \gamma^{(g_m)} \int_{\mu_l^{(g_m)} - \frac{\Delta_l^{(g_m)}}{2}}^{\mu_l^{(g_m)} + \frac{\Delta_l^{(g_m)}}{2}} \rho_l^{(g_m)}(\theta) \textbf{u}(\theta) \textbf{u}(\theta)^H d\theta, \label{R_l}
\end{equation}

\noindent
where $ \sqrt{\gamma^{(g_m)}} $ is the channel gain of user $ m $ in group $ g $ satisfying $ \sum_{l=0}^{L-1} \mathrm{Tr} \big\{\textbf{R}_l^{(g_m)} \big\} =  \gamma^{(g_m)}$ relation. For $ l^{th} $ MPC of user $ m $ in group $ g $, angular power profile function, mean angle of arrival (AoA) and angular spread (AS) are denoted by $ \rho_l^{(g_m)}(\theta) $, $ \mu_l^{(g_m)} $ and $ \Delta_l^{(g_m)} $, respectively. Unit norm steering vector $ \textbf{u}(\theta) \in \mathbb{C}^{M \times 1} $ of uniform linear array (ULA) with half the wavelength spacing corresponding to azimuth angle $ \theta $ is defined as

\begin{equation}
\textbf{u}(\theta) \triangleq \frac{1}{\sqrt{M}} \big[1 \; e^{j \pi sin(\theta)} \cdots \: e^{j (M-1) \pi sin(\theta)} \big]^T. \label{u_theta}
\end{equation}

Although we considered ULA, formulations in this paper can be used for other array structures. The BS observes a sparse angle-delay profile at mm-wave frequencies \cite{SPmmWave}. That is, users have a few active MPCs with narrow angular spread which makes user-grouping feasible. In this paper, it is assumed that users in the same group has the same active MPCs with similar mean AoAs constituting MPC clusters.

\subsubsection{Two-Stage Beamforming Concept}

\begin{figure*}[!h]
\centering
\includegraphics[width=1\textwidth]{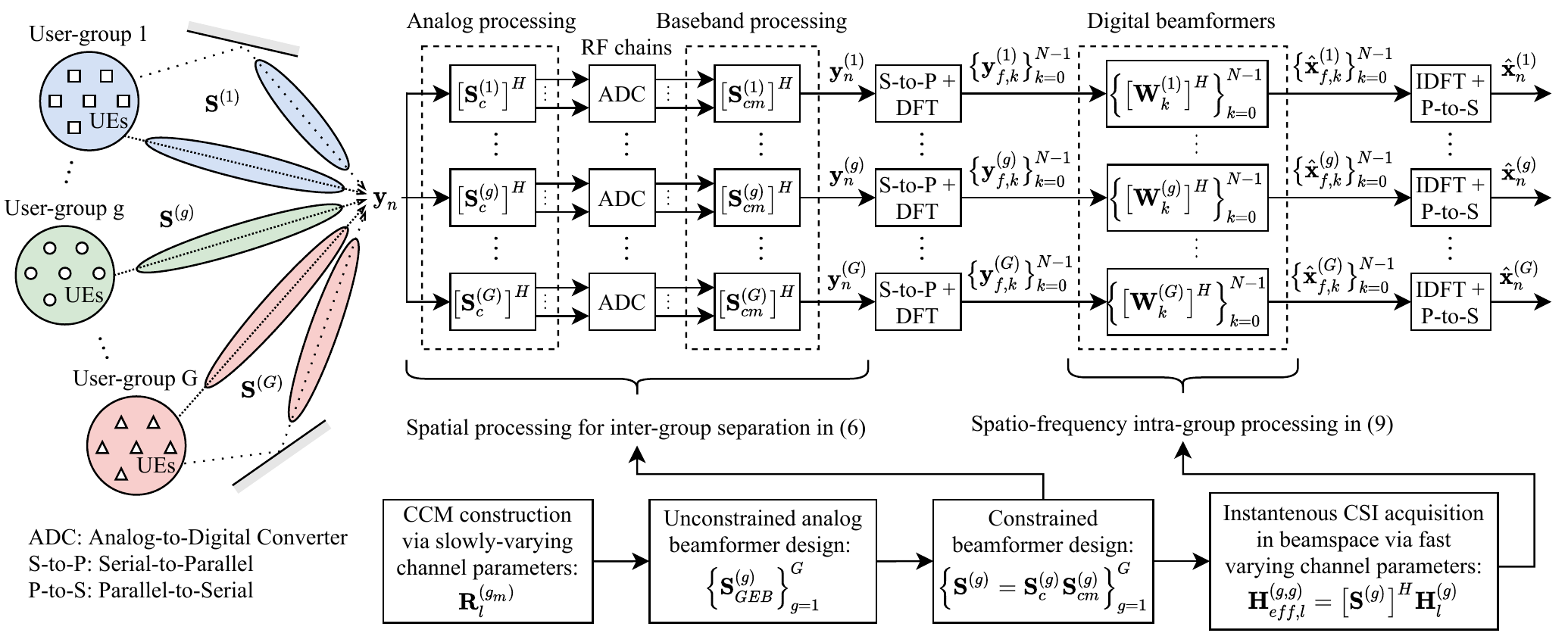}
\caption{Block diagram of the overall system design with constrained analog beamformers}
\label{BlockDiagram}
\vspace{-5mm}
\end{figure*}

This paper considers a two-stage beamforming architecture based on JSDM where there is an analog beamformer $ \textbf{S}^{(g)} \in \mathbb{C}^{M \times D_g} $ for each group suppressing channels of users in other groups. Let the received signal in \eqref{y_n} is rewritten as

\begin{equation}
\textbf{y}_n = \underbrace{ \sum_{l=0}^{L-1} \textbf{H}_{l}^{(g)} \textbf{x}_{(n-l)_N}^{(g)} }_{ \textbf{s}_n^{(g)} \text{: Intra-Group Signals} } + \underbrace{ \sum_{g' \neq g} \sum_{l=0}^{L-1} \textbf{H}_{l}^{(g')} \textbf{x}_{(n-l)_N}^{(g')} + \textbf{n}_n }_{ \boldsymbol{\eta}_n^{(g)} \text{: Inter-Group Signals + Noise Terms} }, \label{y_n2}
\end{equation}

\noindent
where intra-group and inter-group signals are separated for group $ g $. Then, signal vector of group $ g $ after its analog beamformer can be expressed as

\begin{multline}\label{y_n_g}
\tilde{\textbf{y}}_n^{(g)} = \big[ \textbf{S}^{(g)} \big]^H \textbf{y}_n = \underbrace{  \sum_{l=0}^{L-1} \textbf{H}_{eff,l}^{(g,g)} \textbf{x}_{(n-l)_N}^{(g)} }_{ \tilde{\textbf{s}}_{n}^{(g)} \text{: Intra-Group Signals in RD} } \\ + \underbrace{ \sum_{g' \neq g} \sum_{l=0}^{L-1} \textbf{H}_{eff,l}^{(g,g')} \textbf{x}_{(n-l)_N}^{(g')} + \big[ \textbf{S}^{(g)} \big]^H \textbf{n}_n }_{ \tilde{\boldsymbol{\eta}}_{n}^{(g)} \text{: Inter-Group Signals + Noise Terms in RD} } ,
\end{multline}

\noindent
where effective channel matrices at $ l^{th} $ MPC are denoted by $ \textbf{H}_{eff,l}^{(g,g')} \triangleq \big[ \textbf{S}^{(g)} \big]^H \textbf{H}_{l}^{(g')} \in \mathbb{C}^{D_g \times K_{g'}} $. Intra-group signals and sum of inter-group signals and noise in reduced dimension (RD) are denoted by $ \tilde{\textbf{s}}_{n}^{(g)} = \big[ \textbf{S}^{(g)} \big]^H \textbf{s}_n^{(g)} \in \mathbb{C}^{D_g \times 1} $ and $ \tilde{\boldsymbol{\eta}}_{n}^{(g)} = \big[ \textbf{S}^{(g)} \big]^H \boldsymbol{\eta}_n^{(g)} \in \mathbb{C}^{D_g \times 1} $, respectively.

In this paper, frequency domain equalization (FDE) is employed to mitigate ISI caused by SC transmission. Since FDE is employed, DFT operation is applied to the sequence $ \big\{ \tilde{\textbf{y}}_n^{(g)} \big\}_{n=0}^{N-1} $. Using the circularity in \eqref{y_n_g}, one can obtain the signal at the $ k^{th} $ frequency bin after the DFT operation as

\begin{equation}
\tilde{\textbf{y}}_{f,k}^{(g)} = \underbrace{ \textbf{\Lambda}_{eff,k}^{(g,g)} \textbf{x}_{f,k}^{(g)} }_{ \tilde{\textbf{s}}_{f,k}^{(g)} } + \underbrace{ \sum_{g' \neq g} \textbf{\Lambda}_{eff,k}^{(g,g')} \textbf{x}_{f,k}^{(g')} + \tilde{\textbf{n}}_{f,k}^{(g)} }_{ \tilde{\boldsymbol{\eta}}_{f,k}^{(g)} }, \label{Y_k_g}
\end{equation}

\noindent
for $ k = 0,1,\dots,N-1 $, where the DFT of effective channel matrices are denoted by $ \textbf{\Lambda}_{eff,k}^{(g,g')} = \sum_{l=0}^{L-1} \textbf{H}_{eff,l}^{(g,g')} e^{-j \frac{2 \pi}{N} kl} $. In \eqref{Y_k_g}, $ \tilde{\textbf{y}}_{f,k}^{(g)} $, $ \tilde{\textbf{s}}_{f,k}^{(g)} $, $ \textbf{x}_{f,k}^{(g)} $, $ \tilde{\boldsymbol{\eta}}_{f,k}^{(g)} $ and $ \tilde{\textbf{n}}_{f,k}^{(g)} $ are normalized DFTs of vector sequences $ \tilde{\textbf{y}}_n^{(g)} $, $ \tilde{\textbf{s}}_n^{(g)} $, $ \textbf{x}_n^{(g)} $, $ \tilde{\boldsymbol{\eta}}_{n}^{(g)} $ and $ \tilde{\textbf{n}}_n^{(g)} = \big[ \textbf{S}^{(g)} \big]^H \textbf{n}_n $, respectively. Note that normalized DFT of any vector sequence $ \{ \textbf{v_n} \}_{n=0}^{N-1} $ can be calculated as

\begin{equation}
\textbf{v}_{f,k} \triangleq \frac{1}{\sqrt{N}} \sum_{n=0}^{N-1} \textbf{v}_n e^{-j \frac{2 \pi}{N} kn}, \quad k = 0,1,\dots,N-1. \label{V_k}
\end{equation}

Digital beamformer is used to separate intra-group signals of users in JSDM framework. Since we employed FDE structure, digital beamformer is applied in frequency domain. Let $ \textbf{W}_k^{(g)} $ be the digital beamformer of group $ g $ for $ k^{th} $ frequency bin. Then, estimate of $ \textbf{x}_{f,k}^{(g)} $ which is found by using \eqref{Y_k_g} and inverse DFT of the estimates are expressed as follows:

\begin{equation}
\hat{\textbf{x}}_{f,k}^{(g)} = \big[ \textbf{W}_k^{(g)} \big]^H \tilde{\textbf{y}}_{f,k}^{(g)}, \label{X_k_g}
\end{equation}

\begin{equation}
\hat{\textbf{x}}_n^{(g)} = \frac{1}{\sqrt{N}} \sum_{k = 0}^{N-1} \hat{\textbf{x}}_{f,k}^{(g)} e^{j \frac{2 \pi}{N} kn} \label{x_n_g}
\end{equation}

\noindent
for $ k,n = 0,1,\dots,N-1 $. Finally, time-domain estimates $ \hat{\textbf{x}}_n^{(g)} \triangleq \big[ \hat{x}_n^{(g_1)}, \dots, \hat{x}_n^{(g_{K_g})} \big]^T $ are used to demodulate the transmitted symbols, where $ \hat{x}_n^{(g_m)} $ is the estimate of $ x_n^{(g_m)} $. Block diagram summarizing the overall two-stage beamforming architecture described in this section is given in Fig. \ref{BlockDiagram}. Note that analog beamforming stage is given for the constrained beamformer design that will be introduced in Section \ref{sec:constrained}.

\section{Unconstrained Hybrid Massive MIMO Beamformer Design}

In JSDM framework, analog beamformers are designed by using CCMs that are assumed to be known perfectly in this paper\footnote{CCM estimation is out of scope of this paper. However, there are several papers proposing efficient methods for CCM estimation \cite{CCM_Caire1,CCM_Gao,CCM_Caire2,CCM_GMG}. It can be assumed that CCMs can be learned by using the methods mentioned in these papers. It is important to note that CCMs are slowly varying parameters which is the reason why they do not need to be updated frequently. Hence, CCM estimation overhead does not cause a significant increase in overall system complexity.}. Assuming inter-group interference is completely mitigated at analog beamforming stage, digital processing of each group is implemented independently. In this section, analog and digital beamformer designs are considered. 

\subsection{Optimal Analog Beamformer Design: Generalized Eigenbeamformer (GEB)}
\label{sec:analog}

In this section, we find a near-optimal unconstrained analog beamformer structure for JSDM framework where frequency-selectivity of channels are considered. We approach the analog beamformer design as a dimension reduction problem. That is, we aim to find a near-optimal unconstrained analog beamformer, which results in a linear transformation, that preserves the mutual information in reduced dimension. The cost function that needs to be maximized in reduced dimension is the mutual information between $ \tilde{\textbf{y}}_{f,k}^{(g)} $ and $ \tilde{\textbf{s}}_{f,k}^{(g)} $ in \eqref{Y_k_g} which can be expressed as

\begin{equation}
I \Big( \tilde{\textbf{s}}_{f,k}^{(g)}; \tilde{\textbf{y}}_{f,k}^{(g)} \Big) = \mathrm{log}_2 \bigg( \Big\lvert \textbf{I}_{D_g} + \Big[ \textbf{R}_{\tilde{\boldsymbol{\eta}}_{f}}^{(g)} \Big]^{-1} \textbf{R}_{\tilde{\textbf{s}}_{f}}^{(g)} \Big\rvert \bigg), \label{I}
\end{equation}

\noindent
if $ \tilde{\textbf{y}}_{f,k}^{(g)} $ and $ \tilde{\textbf{s}}_{f,k}^{(g)} $ are assumed to be Gaussian random variables. Covariance matrices of $ \tilde{\textbf{s}}_{f,k}^{(g)} $ and $ \tilde{\boldsymbol{\eta}}_{f,k}^{(g)} $ are denoted by $ \textbf{R}_{\tilde{\textbf{s}}_{f}}^{(g)} = \mathbb{E} \Big\{ \tilde{\textbf{s}}_{f,k}^{(g)} \big[ \tilde{\textbf{s}}_{f,k}^{(g)} \big]^H \Big\} $ and $ \textbf{R}_{\tilde{\boldsymbol{\eta}}_{f}}^{(g)} = \mathbb{E} \Big\{ \tilde{\boldsymbol{\eta}}_{f,k} \big[ \tilde{\boldsymbol{\eta}}_{f,k} \big]^H \Big\} $ for $ k = 0,1,\dots,N-1 $, respectively. Note that these covariance matrices are independent of frequency bin index $ k $ according to Appendix \ref{cov_freq}\footnote{In this paper, although frequency-wideband effect is taken into account, spatial-wideband (beam-squint) effect is ignored as bandwidth is assumed to be not that wide \cite{SpatialWideband}. If a spatial-wideband system was considered, spatial covariance matrices of signals would depend on the frequency bin index. However, proposed analog beamforming could still be used either by using the mean spatial covariance matrices, which are averaged over frequency bins, to find the unconstrained analog beamformer as in \cite{DynamicWideband,HybridBeamSquint} (which do not consider the interference rejection capability of pre-beamformer) or by finding the unconstrained analog beamformer with spatial-narrowband assumption and applying a beam-squint compensation in digital domain \cite{BeamSquintComp}.}. Hence, there will be a single cost function to maximize. Furthermore, it is easy to see that maximization of mutual information is equivalent to maximization of the determinant inside the logarithm in \eqref{I}. If we replace $ \textbf{R}_{\tilde{\textbf{s}}_{f}}^{(g)} $ and $ \textbf{R}_{\tilde{\boldsymbol{\eta}}_{f}}^{(g)} $ with their expressions given in Appendix \ref{cov_freq}, the cost function that needs to be maximized under the constraint that $ \textbf{S}^{(g)} $ is a full column matrix can be written as

\begin{equation}
\begin{aligned}
& \underset{ \textbf{S}^{(g)} }{\text{maximize}}
& & \! \! \! \big\lvert \textbf{I}_{D_g} \! \! + \!\big( [ \textbf{S}^{(g)} ]^H \textbf{R}_{\boldsymbol{\eta}}^{(g)} \textbf{S}^{(g)} \big)^{-1} \big( [ \textbf{S}^{(g)} ]^H \textbf{R}_{\textbf{s}}^{(g)} \textbf{S}^{(g)} \big) \big\rvert \\
& \text{subject to}
& & \! \! \! \text{rank} \big( \textbf{S}^{(g)} \big) = D_g
\end{aligned}  \label{det}
\end{equation}

\noindent
where covariance matrix of intra-group signals and covariance matrix of sum of inter-group signals and noise terms in \eqref{y_n2} are expressed as $ \mathbb{E} \Big\{ \textbf{s}_{n}^{(g)} \big[ \textbf{s}_{n'}^{(g)} \big]^H \Big\} = \textbf{R}_{\textbf{s}}^{(g)} \delta_{nn'} $ and $ \mathbb{E} \Big\{ \boldsymbol{\eta}_{n}^{(g)} \big[ \boldsymbol{\eta}_{n'}^{(g)} \big]^H \Big\} = \textbf{R}_{\boldsymbol{\eta}}^{(g)} \delta_{nn'} $, respectively. Considering the uncorrelated nature of transmitted symbols among users and time in addition to the uncorrelated channel vectors in \eqref{E(h_l)}, these two covariance matrices are computed as

\begin{equation}
\textbf{R}_{\textbf{s}}^{(g)} = \frac{E_s^{(g)}}{K_{g}} \sum_{m=1}^{K_g} \sum_{l=1}^{L-1} \textbf{R}_l^{(g_m)}, \label{R_s}
\end{equation}

\begin{equation}
\textbf{R}_{\boldsymbol{\eta}}^{(g)} = \sum_{g' \neq g} \frac{E_s^{(g')}}{K_{g'}} \sum_{m=1}^{K_{g'}} \sum_{l=1}^{L-1} \textbf{R}_l^{(g'_m)} + N_0 \textbf{I}_{M}. \label{R_eta}
\end{equation}

Analog beamformer $ \bar{\textbf{S}}^{(g)} $ that maximizes the cost function is found as \cite{GMG}

\begin{equation}
\bar{\textbf{S}}^{(g)} = \big[ \textbf{v}_1, \textbf{v}_2, \cdots, \textbf{v}_{D_g} \big],\label{S_g}
\end{equation}

\noindent
where $ \textbf{v}_i $ is the eigenvector corresponding to the $ i^{th} $ most dominant eigenvalue (i.e., $ \lambda_i $) of the generalized eigenvalue problem which is defined as

\begin{equation}
\textbf{R}_{\textbf{s}}^{(g)} \textbf{v}_i = \lambda_i \textbf{R}_{\boldsymbol{\eta}}^{(g)} \textbf{v}_i, \quad i = 1,2,\dots,M.\label{GEB}
\end{equation}

This solution is known as generalized eigenbeamformer which creates deep nulls for interfering signals. It is meaningful to use $ \textbf{R}_{\textbf{s}}^{(g)} $, which covers CCMs of all MPCs of intended group, in the generalized eigenvalue problem in \eqref{GEB} since SC-FDE is employed. In order to obtain orthonormalized analog beamformers, QR decomposition of GEB in \eqref{S_g} is found as $ \bar{\textbf{S}}^{(g)} = \textbf{S}^{(g)} \bar{\textbf{R}}^{(g)} $ which is permissible according to Lemma \ref{LEM1}. Equivalent analog beamformer $ \textbf{S}^{(g)} $ found by QR decomposition is used as GEB for the rest of the paper.

\begin{lem}
Using $ \textbf{S}^{(g)} \textbf{A} $ instead of $ \textbf{S}^{(g)} $ does not change the cost function where $ \textbf{A} \in \mathbb{C}^{D_g \times D_g} $ is any invertible matrix. \label{LEM1}
\end{lem}
\begin{proof}
See Appendix \ref{proof_lem1}.
\end{proof}

\subsection{Digital Beamformer Design}

As digital processing of each user group is independent from each other, only intra-group effective channels of users are required for digital beamformer design. Effective channel vectors are learned in time division duplexing (TDD) mode as in Section \ref{sec:ChannelEstimation}. Intra-group channel estimation in reduced dimension reduces the channel estimation overhead significantly. In this section, ZF and LMMSE type digital beamformers are introduced.

\subsubsection{ZF Type Digital Beamformer}

Digital beamformers can be designed in several different ways. If analog beamformers are properly designed, inter-group interference terms in \eqref{Y_k_g} will be close to zero. In this case, classical zero-forcing (ZF) approach can be adopted. ZF type digital beamformer of group $ g $ is expressed as

\begin{equation}
\textbf{W}_k^{(g)} = \boldsymbol{\Lambda}_{eff,k}^{(g,g)} \Big( \big[ \boldsymbol{\Lambda}_{eff,k}^{(g,g)} \big]^H \boldsymbol{\Lambda}_{eff,k}^{(g,g)} \Big)^{-1}. \label{ZF}
\end{equation}

\subsubsection{LMMSE Type Digital Beamformer}

Constrained beamformers that will be introduced in Section \ref{sec:constrained} do not necessarily suppress inter-group interference. In this case, ZF type digital beamformers could perform poorly even at high signal-to-noise ratio (SNR) regime. Thus, we adopt LMMSE type beamformer to suppress residual inter-group interference in the reduced dimension. LMMSE type beamformer of group $ g $ is expressed as

\begin{equation}
\textbf{W}_k^{(g)} = \textbf{R}_{\tilde{\textbf{y}}_{f,k}^{(g)}}^{-1} \textbf{R}_{\tilde{\textbf{y}}_{f,k}^{(g)} \textbf{x}_{f,k}^{(g)}}, \label{MMSE}
\end{equation}

\noindent
where covariance matrix of $ \tilde{\textbf{y}}_{f,k}^{(g)} $ and covariance matrix between $ \tilde{\textbf{y}}_{f,k}^{(g)} $ and $ \textbf{x}_{f,k}^{(g)} $ are denoted by $ \textbf{R}_{\tilde{\textbf{y}}_{f,k}^{(g)}} $ and $ \textbf{R}_{\tilde{\textbf{y}}_{f,k}^{(g)} \textbf{x}_{f,k}^{(g)}} $, respectively. Using the fact that transmitted symbol vectors of users are uncorrelated and independent from the noise vector, these covariance matrices are computed as

\begin{equation}
\textbf{R}_{\tilde{\textbf{y}}_{f,k}^{(g)} \textbf{x}_{f,k}^{(g)}} = \mathbb{E} \Big\{ \tilde{\textbf{y}}_{f,k}^{(g)} \big[ \textbf{x}_{f,k}^{(g)} \big]^H \Big\} = \frac{E_s^{(g)}}{K_g} \boldsymbol{\Lambda}_{eff,k}^{(g,g)},
\label{R_Y_k_X_k}
\end{equation}

\begin{equation}
\textbf{R}_{\tilde{\textbf{y}}_{f,k}^{(g)}} = \mathbb{E} \Big\{ \tilde{\textbf{y}}_{f,k}^{(g)} \big[ \tilde{\textbf{y}}_{f,k}^{(g)} \big]^H \Big\} = \frac{E_s^{(g)}}{K_g} \boldsymbol{\Lambda}_{eff,k}^{(g,g)} \big[ \boldsymbol{\Lambda}_{eff,k}^{(g,g)} \big]^H \! \! + \textbf{R}_{ \tilde{\boldsymbol{\eta}}_{f}}^{(g)}, \label{R_Y_k}
\end{equation}

\noindent
for given frequency-domain intra-group channel matrix $ \boldsymbol{\Lambda}_{eff,k}^{(g,g)} $. Recall that intra-group channels are available at the BS whereas inter-group channel matrices $ \boldsymbol{\Lambda}_{eff,k}^{(g,g')} $ are not known. That is why inter-group channel matrices are treated as random matrices in \eqref{R_Y_k} and covariance matrix of sum of inter-group interference and noise terms in reduced dimension $ \textbf{R}_{ \tilde{\boldsymbol{\eta}}_{f}}^{(g)} $ whose expression is given in Appendix \ref{cov_freq} is used. Consequently, instantaneous channel estimation overhead is the same for both ZF and LMMSE type digital beamformers.

\section{Constrained Analog Beamformer Design}
\label{sec:constrained}

In this paper, optimality is imposed at the analog beamforming stage. Near-optimal GEB introduced in Section \ref{sec:analog} is an unconstrained beamformer, i.e., it does not obey constant-modulus constraint. However, phase-shifter networks are preferred for the analog beamforming stage in practice \cite{SPmmWave}. Hence, our goal is to find a constant-modulus approximation to near-optimal GEB of each group. In addition, we provide a DFT beamformer structure which is obtained by using AoAs belonging to MPC clusters of user groups. In this section, near-optimal GEB of group $ g $ found in \eqref{S_g} in Section \ref{sec:analog} and its constant-modulus constrained equivalent are denoted by $ \textbf{S}_{GEB}^{(g)} \in \mathbb{C}^{M \times D_g}$ and $ \textbf{S}_{c}^{(g)} \in \mathbb{C}^{M \times D_g} $, respectively.

\subsection{Fully Connected Array}

\subsubsection{DFT Beamformer}

DFT beamformer, which will be used for comparison purposes, is widely used in the literature \cite{JSDM,DFT,TwoStage}. It is not an interference aware beamformer. In other words, it only considers the AoAs of MPC clusters in the intended group and does not consider the channels or signals of other groups. Let $ \textbf{Q} \in \mathbb{C}^{M \times M} $ be the normalized DFT matrix with entries $ [ \textbf{Q} ]_{(m,n)} \triangleq \frac{1}{\sqrt{M}} e^{-j \frac{2\pi}{M}\pi mn} $ for $ m,n = 0,1,\dots,M-1 $. Then, DFT beamformer of group $ g $ is defined as

\begin{equation}
\textbf{S}_{c}^{(g)} \triangleq \textbf{Q} \boldsymbol{\Gamma}^{(g)}, \label{DFT}
\end{equation}

\noindent
where $ \boldsymbol{\Gamma}^{(g)} = [\textbf{e}_{g,1}, \textbf{e}_{g,2}, \dots, \textbf{e}_{g,{D_g}}] \in \mathbb{Z}^{M \times D_g} $ is used as the column selection matrix for group $ g $. Columns of $ \boldsymbol{\Gamma}^{(g)} $ are denoted by $ \textbf{e}_{i} \in \mathbb{Z}^{M \times 1} $ which is a binary vector whose $ i^{th} $ entry is one whereas others are zero. This vector selects the columns of the DFT matrix for the analog beamformer. For any group, firstly, closest DFT vectors to mean AoAs of active MPCs are selected. If the number of available RF chains is greater than the number of active MPCs, which is the desired case, remaining columns are selected from the adjacent DFT vectors to the ones closest to mean AoAs. Analog beamformer used in practice is expressed as $ \textbf{S}^{(g)} = \textbf{S}_{c}^{(g)} $.

\subsubsection{Phase Extraction (PE)}

Constrained analog beamformer $ \textbf{S}_{c}^{(g)} $ has only constant-modulus constraint for the fully connected structure, i.e., RF chains are connected to all the antennas. Since our goal is to find the best approximation to GEB, an optimization problem can be constructed as follows:

\begin{equation}\label{problem1}
\begin{aligned}
& \underset{\textbf{S}_{c}^{(g)}}{\text{minimize}}
& & \norm{\textbf{S}_{GEB}^{(g)} - \textbf{S}_{c}^{(g)} }_F^2 \\
& \text{subject to}
& & \big| [\textbf{S}_{c}^{(g)}]_{(i,j)} \big| = \frac{1}{\sqrt{M}}, \forall i,j,
\end{aligned}
\end{equation}

\noindent
which has the simple phase extraction solution where optimum constrained analog beamformer can be expressed as $ \textbf{S}_{c}^{(g)\star}  = \frac{1}{\sqrt{M}} e^{j\measuredangle \big( \textbf{S}_{GEB}^{(g)} \big)} $. This solution is used in the literature due to its simplicity. If multiuser MIMO without user-grouping is considered, phase extraction offers acceptable performance \cite{Exploiting,Dynamic}. However, deep nulls introduced by GEB may disappear which results in increase in inter-group interference for JSDM. Similar to DFT beamformer, practical analog beamformer is expressed as $ \textbf{S}^{(g)} = \textbf{S}_{c}^{(g)} $.

\subsubsection{Phase Extraction with Alternating Minimization (PE-AM)}

In order to improve the accuracy of approximation, we introduce a compensation matrix $ \textbf{S}_{cm}^{(g)} \in \mathbb{C}^{D_g \times D_g} $ which will be applied in the digital baseband after the constrained analog beamformer $ \textbf{S}_{c}^{(g)} $. This can be considered as dividing the analog beamformer into two stages. A new optimization problem with the compensation matrix can be constructed as

\begin{equation}\label{problem2}
\begin{aligned}
& \underset{\textbf{S}_{c}^{(g)}, \textbf{S}_{cm}^{(g)}}{\text{minimize}}
& & \norm{\textbf{S}_{GEB}^{(g)} - \textbf{S}_{c}^{(g)} \textbf{S}_{cm}^{(g)}}_F^2 \\
& \text{subject to}
& & \big| [\textbf{S}_{c}^{(g)}]_{(i,j)} \big| = \frac{1}{\sqrt{M}}, \forall i,j.
\end{aligned}
\end{equation}

Optimization problem given above can be solved with manifold optimization-based algorithms which have high complexity \cite{AltMin}. Columns of near-optimal GEB are mutually orthonormal whereas there is no constraint on the compensation matrix. Although optimum structure of the compensation matrix is not known, orthogonality constraint on compensation matrix is imposed to obtain a simplified optimization problem. It can be shown that the cost function in \eqref{problem2} is upper-bounded by $ \norm{\textbf{S}_{GEB}^{(g)} \big[ \textbf{S}_{cm}^{(g)} \big]^H - \textbf{S}_{c}^{(g)}}_F^2 $ due to the orthogonality constraint of the compensation matrix. Modified optimization problem with this upper bound is expressed as

\begin{equation}\label{problem3}
\begin{aligned}
& \underset{\textbf{S}_{c}^{(g)}, \textbf{S}_{cm}^{(g)}}{\text{minimize}}
& & \norm{\textbf{S}_{GEB}^{(g)} \big[ \textbf{S}_{cm}^{(g)} \big]^H - \textbf{S}_{c}^{(g)}}_F^2 \\
& \text{subject to}
& & \big| [\textbf{S}_{c}^{(g)}]_{(i,j)} \big| = \frac{1}{\sqrt{M}}, \forall i,j; \quad \! \! \! \! \big[ \textbf{S}_{cm}^{(g)} \big]^H \textbf{S}_{cm}^{(g)} = \textbf{I}_{D_g}.
\end{aligned}
\end{equation}

Note that $ \textbf{S}_{GEB}^{(g)} \big[ \textbf{S}_{cm}^{(g)} \big]^H $ is an equivalent generalized eigenbeamformer according to Lemma \ref{LEM1} since $ \big[ \textbf{S}_{cm}^{(g)} \big]^H $ is invertible. Given $ \textbf{S}_{cm}^{(g)} $, optimum solution for the constrained analog beamformer $ \textbf{S}_{c}^{(g)\star} $ is a phase extraction of the equivalent beamformer that is expressed as

\begin{equation}
\textbf{S}_{c}^{(g)\star} = \frac{1}{\sqrt{M}} e^{j \measuredangle \big( \textbf{S}_{GEB}^{(g)} [ \textbf{S}_{cm}^{(g)} ]^H \big)}. \label{Full_S_c}
\end{equation}

For given constrained analog beamformer $ \textbf{S}_{c}^{(g)} $, optimization problem given in \eqref{problem3} reduces to

\begin{equation}\label{problem4}
\begin{aligned}
& \underset{\textbf{S}_{cm}^{(g)}}{\text{minimize}}
& & \norm{\textbf{S}_{GEB}^{(g)} \big[ \textbf{S}_{cm}^{(g)} \big]^H - \textbf{S}_{c}^{(g)}}_F^2 \\
& \text{subject to}
& & \big[ \textbf{S}_{cm}^{(g)} \big]^H \textbf{S}_{cm}^{(g)} = \textbf{I}_{D_g}.
\end{aligned}
\end{equation}

\noindent
which is similar to orthogonal Procrustes problem \cite{OPP}. Optimal compensation matrix can be found as

\begin{equation}
\textbf{S}_{cm}^{(g)\star} = \textbf{V}^{(g)} \big[ \textbf{U}^{(g)} \big]^H, \label{Full_S_cm}
\end{equation}

\noindent
where following singular value decomposition (SVD) is used, $ \big[ \textbf{S}_{GEB}^{(g)} \big]^H \textbf{S}_{c}^{(g)} = \textbf{U}^{(g)} \boldsymbol{\Sigma}^{(g)} \big[ \textbf{V}^{(g)} \big]^H $ \cite{AltMin}. Alternating minimization algorithm is adopted in order to obtain optimum solutions where $ \textbf{S}_{c}^{(g)\star} $ is calculated for given $ \textbf{S}_{cm}^{(g)} $ and $ \textbf{S}_{cm}^{(g)\star} $ is calculated by using given $ \textbf{S}_{c}^{(g)} $ repeatedly until convergence is achieved. Summary of the described solution is given in Algorithm \ref{PE-AM}.

\begin{algorithm}
\caption{PE-AM}\label{PE-AM}
\begin{algorithmic}[1]

\Require $ \textbf{S}_{GEB}^{(g)} $
\Ensure $ \textbf{S}_{c,0}^{(g)} = \frac{1}{\sqrt{M}} e^{ j \measuredangle \big( \textbf{S}_{GEB}^{(g)} \big) } $, iteration index $ n = 0 $
   
	\Repeat
	
		\State Fix $ \textbf{S}_{c,n}^{(g)} $, find $ \big[ \textbf{S}_{GEB}^{(g)} \big]^H \textbf{S}_{c,n}^{(g)} = \textbf{U}_{n}^{(g)} \boldsymbol{\Sigma}_{n}^{(g)} \big[ \textbf{V}_{n}^{(g)} \big]^H $
		
		\State $ \textbf{S}_{cm,n}^{(g)} = \textbf{V}_{n}^{(g)} \big[ \textbf{U}_{n}^{(g)} \big]^H $	

		\State Fix $ \textbf{S}_{cm,n}^{(g)} $, find $ \textbf{S}_{c,n+1}^{(g)} = \frac{1}{\sqrt{M}} e^{ j \measuredangle \Big( \textbf{S}_{GEB}^{(g)} [ \textbf{S}_{cm,n}^{(g)} ]^H \Big) }  $		

		\State $ n \gets n + 1 $
	
	\Until{$ \norm{\textbf{S}_{GEB}^{(g)} \big[ \textbf{S}_{cm,n}^{(g)} \big]^H - \textbf{S}_{c,n+1}^{(g)}}_F $ converges}

\Output $ \textbf{S}_{c}^{(g)} = \textbf{S}_{c,n+1}^{(g)}  $ and $ \textbf{S}_{cm}^{(g)} = \textbf{S}_{cm,n}^{(g)} $

\end{algorithmic}
\end{algorithm}

For compensation matrix-based design, constrained beamformer $ \textbf{S}_c^{(g)} $ is applied in analog domain whereas compensation matrix $ \textbf{S}_{cm}^{(g)} $ is applied in digital baseband just before the digital beamformers. Thus, overall analog beamformer of group $ g $ becomes $ \textbf{S}^{(g)} = \textbf{S}_c^{(g)} \textbf{S}_{cm}^{(g)} $ as it is seen from Fig. \ref{BlockDiagram}.

\subsection{Partially Connected Array}

Unlike fully connected arrays, each antenna element is connected to only one RF chain for partially connected designs. In other words, analog beamformer has a sparse structure resulting in less phase shifters. Let the binary connection matrix of group $ g $ is denoted by $ \textbf{\Pi}^{(g)} \in \mathbb{Z}^{M \times D_g} $ where each row contains only one non-zero entry. Column where $ i^{th} $ row has the non-zero entry is denoted by $ j(i) $. Constrained beamforming matrix $ \textbf{S}_{c}^{(g)} $ has non-zero entries where the connection matrix $ \textbf{\Pi}^{(g)} $ has non-zero entries. Similar to PE-AM algorithm, we consider joint optimization of constrained beamformer and compensation matrix. There are two types of partially connected arrays. In the first one, optimization problem is solved by using a fixed connection matrix. The second one considers finding the optimum connection matrix for a given unconstrained beamformer.

\subsubsection{Fixed Subarray Design}

Joint optimization problem of constrained analog beamformer with partially connected array and compensation matrix can be written as

\begin{equation}\label{problem5}
\begin{aligned}
& \underset{\textbf{S}_{c}^{(g)}, \textbf{S}_{cm}^{(g)}}{\text{minimize}}
& & \norm{\textbf{S}_{GEB}^{(g)} - \textbf{S}_{c}^{(g)} \textbf{S}_{cm}^{(g)}}_F^2 \\
& \text{subject to}
& & \big| [\textbf{S}_{c}^{(g)}]_{(i,j)} \big| = \frac{1}{\sqrt{M}} [\textbf{\Pi}^{(g)}]_{(i,j)}, \forall i,j,
\end{aligned}
\end{equation}

\noindent
for a given connection matrix $ \textbf{\Pi}^{(g)} $. Note that since there is only one non-zero element in each row of $ \textbf{S}_{c}^{(g)} $, non-zero element in $ i^{th} $ row is multiplied with $ j(i)^{th} $ row of the compensation matrix $ \textbf{S}_{cm}^{(g)} $. By using this property, optimization problem in \eqref{problem5} is simplified as

\begin{equation}\label{problem6}
\begin{aligned}
& \underset{ \{ \beta_i^{(g)} \}_{i=1}^{M} }{\text{minimize}}
& \norm{\big[ \textbf{S}_{GEB}^{(g)} \big]_{(i,:)} - \frac{1}{\sqrt{M}} e^{j\beta_i^{(g)}} \big[ \textbf{S}_{cm}^{(g)} \big]_{(j(i),:)} }_F^2 \\
\end{aligned}
\end{equation}

\noindent
for given compensation matrix $ \textbf{S}_{cm}^{(g)} $, where $ \beta_i^{(g)} $ is the phase of non-zero entry in the $ i^{th} $ row of $ \textbf{S}_{c}^{(g)} $. It can be seen that optimization of phase values is decoupled. Hence, optimum phase values $ \big\{ \beta_{i}^{(g)\star} \big\}_{i=1}^{M} $ are found as

\begin{equation}
\beta_{i}^{(g)\star} = \measuredangle \Big( \big[ \textbf{S}_{GEB}^{(g)} \big]_{(i,:)} \big[ \textbf{S}_{cm}^{(g)} \big]_{(j(i),:)}^H \Big). \label{S_c_part}
\end{equation}

Given constrained beamformer $ \textbf{S}_{c}^{(g)} $, optimization problem to find the compensation matrix becomes

\begin{equation}\label{problem7}
\begin{aligned}
& \underset{\textbf{S}_{cm}^{(g)}}{\text{minimize}}
& & \norm{\textbf{S}_{GEB}^{(g)} - \textbf{S}_{c}^{(g)} \textbf{S}_{cm}^{(g)}}_F^2. \\
\end{aligned}
\end{equation}

As there is no constraint on the compensation matrix, optimum solution is the least squares (LS) solution given as

\begin{equation}
\textbf{S}_{cm}^{(g)\star} = \Big( \big[ \textbf{S}_{c}^{(g)} \big]^H \textbf{S}_{c}^{(g)} \Big)^{-1} \big[ \textbf{S}_{c}^{(g)} \big]^H \textbf{S}_{GEB}^{(g)}.
 \label{S_cm_part}
\end{equation}

Alternating minimization approach is adopted as in PE-AM algorithm to find constrained analog beamformer and compensation matrix. Summary of the fixed partially connected array design is given in Algorithm \ref{Fixed}.

\begin{algorithm}
\caption{Fixed Subarray Array Design}\label{Fixed}
\begin{algorithmic}[1]

\Require $ \textbf{S}_{GEB}^{(g)}, \textbf{\Pi}^{(g)} $
\Ensure Set random phases to $ \textbf{S}_{c,0}^{(g)} $, iteration index $ n = 0 $
   
	\Repeat
	
		\State Fix $ \textbf{S}_{c,n}^{(g)} $ 
		
		\State Find $ \textbf{S}_{cm,n}^{(g)} = \Big( \big[ \textbf{S}_{c,n}^{(g)} \big]^H \textbf{S}_{c,n} \Big)^{-1} \big[ \textbf{S}_{c,n}^{(g)} \big]^H \textbf{S}_{GEB}^{(g)} $
		
		\State Fix $ \textbf{S}_{cm,n}^{(g)} $, find $ \textbf{S}_{c,n+1}^{(g)} $ by \eqref{S_c_part}	

		\State $ n \gets n + 1 $
	
	\Until{$ \norm{ \textbf{S}_{GEB}^{(g)} - \textbf{S}_{c,n+1}^{(g)} \textbf{S}_{cm,n}^{(g)} }_F $ converges}

\Output $ \textbf{S}_{c}^{(g)} = \textbf{S}_{c,n+1}^{(g)}  $ and $ \textbf{S}_{cm}^{(g)} = \textbf{S}_{cm,n}^{(g)} $

\end{algorithmic}
\end{algorithm}

\vspace{-3mm}
\subsubsection{Dynamic Subarray Design}

This section aims to find an algorithm that finds the optimum connection matrix $ \textbf{\Pi}^{(g)} $ for a given GEB. Let $ \tilde{\textbf{S}}_{c}^{(g)} \in \mathbb{C}^{M \times D_g} $ be a partially connected constrained beamformer whose non-zero entry locations are unknown. This matrix needs to satisfy following constraints

\begin{subequations}
\begin{equation}
\big| [\tilde{\textbf{S}}_{c}^{(g)}]_{(i,j)} \big| \in \{0,1\}, \forall i,j, \label{constraint1}
\end{equation}    
\begin{equation}
\sum_{j=1}^{D} \big| [\tilde{\textbf{S}}_{c}^{(g)}]_{(i,j)} \big| = 1, \forall i,\label{constraint2}
\end{equation}
\begin{equation}
\sum_{i=1}^{M} \big| [\tilde{\textbf{S}}_{c}^{(g)}]_{(i,j)} \big| \geq 1, \forall j,\label{constraint3}
\end{equation}
\end{subequations}

\noindent
where the first one is constant-modulus constraint while second one needs to be satisfied to make sure that each antenna is connected to only one RF chain. The third constraint should be satisfied so that each RF chain is connected to at least one antenna. Let the equivalent near-optimal generalized eigenbeamformer be $ \textbf{S}_{GEB}^{(g)} \textbf{A}^{(g)} $ where $ \textbf{A}^{(g)} \in \mathbb{C}^{D_g \times D_g} $ is an arbitrary unitary matrix. Note that $ \textbf{S}_{GEB}^{(g)} \textbf{A}^{(g)} $ can replace $ \textbf{S}_{GEB}^{(g)} $ as long as $ \textbf{A}^{(g)} $ is invertible according to Lemma \ref{LEM1}. Let us ignore the constraint in \eqref{constraint3} for now and construct an optimization problem given as

\begin{equation}\label{problem8}
\begin{aligned}
& \underset{ \textbf{A}^{(g)} \in \mathcal{A}, \tilde{\textbf{S}}_{c}^{(g)} }{\text{minimize}}
& & \norm{ \textbf{S}_{GEB}^{(g)} \textbf{A}^{(g)} - \tilde{\textbf{S}}_{c}^{(g)} }_F^2 \\
& \text{subject to}
& & \big| [\tilde{\textbf{S}}_{c}^{(g)}]_{(i,j)} \big| \in \{0,1\}, \forall i,j; \\
& \text{}
& & \sum_{j=1}^{D} \big| [\tilde{\textbf{S}}_{c}^{(g)}]_{(i,j)} \big| = 1, \forall i.
\end{aligned}
\end{equation}

\noindent
where $ \mathcal{A} $ is the set of unitary matrices which are invertible. Given $ \textbf{A}^{(g)} $, entries of optimal constrained analog beamformer can be found as

\begin{equation}
[\tilde{\textbf{S}}_{c}^{(g)\star}]_{(i,j)} = 
\begin{dcases}
    e^{j \measuredangle \big( [ \textbf{S}_{GEB}^{(g)} \textbf{A}^{(g)} ]_{(i,j)} \big)}, & \text{if } j = j^{\star}(i)\\
    0 ,& \text{otherwise,}
\end{dcases}\label{S_c_part_tilde}
\end{equation}

\noindent
where $ j^{\star}(i) = \argmax_{1 \leq j \leq  D} | [ \textbf{S}_{GEB}^{(g)} \textbf{A}^{(g)} ]_{(i,j)} |  $ for $ i = 1,2,\dots,M $ \cite{ChStat}. For given $ \tilde{\textbf{S}}_{c}^{(g)} $, optimization problem in \eqref{problem8} reduces to orthogonal Procrustes problem expressed as

\begin{equation}\label{problem9}
\begin{aligned}
& \underset{ \textbf{A}^{(g)} \in \mathcal{A}}{\text{minimize}}
& & \norm{ \textbf{S}_{GEB}^{(g)} \textbf{A}^{(g)} - \tilde{\textbf{S}}_c^{(g)} }_F^2
\end{aligned},
\end{equation}

\noindent
which has the same solution as in the PE-AM algorithm. Optimum unitary matrix is found as

\begin{equation}
\textbf{A}^{(g)\star} = \textbf{V}^{(g)} \big[ \textbf{U}^{(g)} \big]^H, \label{A}
\end{equation}

\noindent
where the following SVD is used: $ \big[ \textbf{S}_{GEB}^{(g)} \big]^H \tilde{\textbf{S}}_{c}^{(g)} = \textbf{U}^{(g)} \boldsymbol{\Sigma}^{(g)} \big[ \textbf{V}^{(g)} \big]^H $. Similar to previous algorithms, alternating minimization is applied to find partially connected constrained beamformer which is summarized in Algorithm \ref{Dynamic1}.

\begin{algorithm}
\caption{Dynamic Connection Design}\label{Dynamic1}
\begin{algorithmic}[1]

\Require $ \textbf{S}_{GEB}^{(g)} $
\Ensure Set random phases to $ \tilde{\textbf{S}}_{c,0}^{(g)} $, iteration index $ n = 0 $
   
	\Repeat
	
		\State Fix $ \tilde{\textbf{S}}_{c,n}^{(g)} $, find $ \big[ \textbf{S}_{GEB}^{(g)} \big]^H \tilde{\textbf{S}}_{c,n}^{(g)} = \textbf{U}_{n}^{(g)} \boldsymbol{\Sigma}_{n}^{(g)} \big[ \textbf{V}_{n}^{(g)} \big]^H $

		\State $ \textbf{A}_{n}^{(g)} = \textbf{V}_{n}^{(g)} \big[ \textbf{U}_{n}^{(g)} \big]^H $
		
		\State Fix $ \textbf{A}_{n}^{(g)} $, find $ \tilde{\textbf{S}}_{c,n+1}^{(g)} $ by \eqref{S_c_part_tilde}	

		\State $ n \gets n + 1 $
	
	\Until{$ \norm{ \textbf{S}_{GEB}^{(g)} \textbf{A}_{n}^{(g)} - \tilde{\textbf{S}}_{c,n+1}^{(g)} }_F $ converges}

\Output $ \tilde{\textbf{S}}_{c}^{(g)} = \tilde{\textbf{S}}_{c,n+1}^{(g)} $

\end{algorithmic}
\end{algorithm}

\begin{rem}
The result of the Algorithm \ref{Dynamic1}, which solves a non-convex problem, depends on the initial random phases which leads to different antenna connections for each result. Hence, this algorithm should be repeated several times with different initial phases and the result yielding maximum expected SINR for group $ g $ in reduced dimension, which is the SINR after statistical pre-beamformer $ \textbf{S}^{(g)} $ (i.e., just before digital beamforming), should be used. Expression for the expected SINR is given in \eqref{SINR_bar} in Appendix \ref{stat_SINR}, where we denote the expected SINR of group $ g $ by $ \overline{SINR}^{(g)} $.\label{rem1}
\end{rem}

\begin{rem}
Constraint \eqref{constraint3} is not considered in Algorithm \ref{Dynamic1} which is repeated several times as stated in Remark \ref{rem1}. Optimum connection matrix should be selected from the results where there is a connection between any RF chain and at least one antenna to satisfy this constraint. \label{rem2}
\end{rem}

It is important to note that only constrained analog beamformer can be found by using Algorithm \ref{Dynamic1}. However, we need to find a compensation matrix to improve the result. We propose that result of Algorithm \ref{Dynamic1} should be used in order to find the connection matrix as $ \textbf{\Pi}^{(g)} = \big| \tilde{\textbf{S}}_{c}^{(g)} \big| $. Then, this connection matrix should be used as an input to Algorithm \ref{Fixed} which jointly optimizes constrained analog beamformer and compensation matrix. Furthermore, initial phases of constrained analog beamformer for Algorithm \ref{Fixed} is selected as the phases of $ \tilde{\textbf{S}}_{c}^{(g)} $. Therefore, it can be seen that Algorithm \ref{Dynamic1} is used to find the connection matrix. Summary of the overall dynamic subarray design is given in Algorithm \ref{Dynamic2} where $ N_{iter} $ is the number of times Algorithm \ref{Dynamic1} is repeated.

\begin{algorithm}
\caption{Dynamic Subarray Design}\label{Dynamic2}
\begin{algorithmic}[1]

\Require $ \textbf{S}_{GEB}^{(g)}, N_{iter} $
\Ensure Iteration index $ t = 0 $

\Repeat

	\State $ t \gets t + 1 $   
   
	\State Run Algorithm \ref{Dynamic1} to find $\tilde{\textbf{S}}_{c}^{(g)} $ and set $ \tilde{\textbf{S}}_{c,t}^{(g)} \! \! \gets \! \tilde{\textbf{S}}_{c}^{(g)} $

	\If {$ \sum_{i=1}^{M} \big| [\tilde{\textbf{S}}_{c,t}^{(g)}]_{(i,j)} \big| \geq 1, \forall j $}
		\State Calculate $ \overline{SINR}_t^{(g)} $ by setting $ \textbf{S}^{(g)} \! \! \gets \! \tilde{\textbf{S}}_{c,t}^{(g)} $ in \eqref{SINR_bar}
	\Else	
		\State Set $ \overline{SINR}_t^{(g)} $ to 0	
	\EndIf

\Until $ t = N_{iter} $

\State $ t^\star = \argmax_{1 \leq t \leq N_{iter}} \overline{SINR}_t^{(g)} $

\State Run Algorithm \ref{Fixed} with $ \textbf{\Pi}^{(g)} = \big| \tilde{\textbf{S}}_{c,t^\star}^{(g)} \big| $ and initial phases of $ \tilde{\textbf{S}}_{c,t^\star}^{(g)} $ to find $ \textbf{S}_{c}^{(g)} $ and $ \textbf{S}_{cm}^{(g)} $.

\Output $ \textbf{S}_{c}^{(g)} $ and $ \textbf{S}_{cm}^{(g)} $

\end{algorithmic}
\end{algorithm}

\subsection{Convergence and Complexity Analysis}

Alternating minimization approach is employed to obtain both fully and partially connected constrained analog beamformers. Convergence of given algorithms is guaranteed since cost functions, which are bounded below, are monotonically decreasing at each iteration \cite{LowComplexityHybrid,ChStat}.

Algorithm 1 and Algorithm 3 has the same solution steps which involve phase extraction of matrix multiplication at one step. Another step includes SVD of a matrix with size $ D_g \times D_g $ which would not impose high complexity as number of RF chains $ D_g $ is limited in hybrid systems. Algorithm 2 only includes phase extraction from a vector multiplication and a LS solution which requires the inverse of a $ D_g \times D_g $ matrix. Similarly, this matrix inverse can be taken with ease as $ D_g $ is limited. Algorithm 4, which is used to find the dynamic constrained analog beamformers, involve execution of Algorithm 3 $ N_{iter} $ times. It can be understood that complexity orders of proposed alternating minimization-based algorithms are not high. Moreover, recall that unconstrained analog beamformers are designed based on slowly varying CCMs. Hence, constrained beamformers and dynamic connections are updated infrequently. To conclude, these algorithms do not increase the overall complexity of the system.

\section{Performance Measures} \label{perf}

In order to show the effectiveness of proposed constrained design, we should define some performance measures. We firstly introduce ergodic capacity which is the most commonly used performance measure. Then, we describe a channel estimation method in reduced dimension and define the normalized mean square error (nMSE) that shows the accuracy of channel estimation.

\subsection{Ergodic Capacity}

Ergodic capacity of user $ m $ in group $ g $ is defined as

\begin{equation}
C^{(g_m)} = \mathbb{E}_{\big\{ \textbf{H}_{eff,l}^{(g,g)} \big\}_{l=0}^{L-1}} \Big\{ \mathrm{log}_2 \Big( 1+SINR^{(g_m)} \Big) \Big| \textbf{S}^{(g)} \Big\},  \label{rate}
\end{equation}

\noindent
where expectation is taken over different instantaneous channel realizations while keeping the analog beamformer constant, i.e., user locations and CCMs are kept constant. Output SINR of user $ m $ in group $ g $, which is the SINR after all the receiver processing, is denoted by $ SINR^{(g_m)} $. Let the Bussgang decomposition of the soft output for decoding transmitted symbols given in \eqref{x_n_g} is expressed as \cite{Bussgang}

\begin{equation}
\hat{x}_n^{(g_m)} = a^{(g_m)} x_n^{(g_m)} + b^{(g_m)}, \label{x_n_g_m}
\end{equation} 

\noindent
for $ n = 0,1,\dots,N-1 $ and $ m = 1,2,\dots,K_g $, where $ a^{(g_m)} $ and $ b^{(g_m)} $ are the complex amplitude and residual interference of the transmitted symbol of user $ m $ in group $ g $, respectively. Then the output SINR of this user for given $ \big\{ \textbf{H}_{eff,l}^{(g)} \big\}_{l=0}^{L-1} $ is calculated as

\begin{equation}
SINR^{(g_m)} = \frac{| a^{(g_m)} |^2 \mathbb{E}\big\{ | x_n^{(g_m)} |^2 \big\} }{ \mathbb{E}\big\{ | b_n^{(g_m)} |^2 \big\}} = \frac{E_s^{(g)}}{K_g} \frac{| a^{(g_m)} |^2}{ \mathbb{E}\big\{ | b_n^{(g_m)} |^2 \big\}}.\label{SINR}
\end{equation}

By using \eqref{Y_k_g}, \eqref{X_k_g} and \eqref{x_n_g}, complex amplitude can be computed as

\begin{equation}
\begin{aligned}
a^{(g_m)} &= \mathbb{E} \big\{ \hat{x}_n^{(g_m)} \big( x_n^{(g_m)} \big)^* \big\} \big/ \mathbb{E}\big\{ | x_n^{(g_m)} |^2 \big\} \\
&= \frac{1}{N} \sum_{k=0}^{N-1} \big[ \textbf{W}_k^{(g)} \big]_{(:,m)}^H \textbf{\Lambda}_{eff,k}^{(g,g)} \textbf{e}_m,
\end{aligned}\label{a^g_m}
\end{equation}

\noindent
where $ \big[ \textbf{W}_k^{(g)} \big]_{(:,m)} $ is the $ m^{th} $ column of $ \textbf{W}_k^{(g)} $ and $ \textbf{e}_m \in \mathbb{Z}^{K_g \times 1} $ is a binary vector whose $ m^{th} $ entry is one whereas others are zero. On the other hand, instead of residual interference, its power can be calculated as $ \mathbb{E} \big\{ | b^{(g_m)} |^2 \big\} = \big\{ | \hat{x}_n^{(g_m)} |^2 \big\} - \frac{E_s^{(g)}}{K_g} | a^{(g_m)} |^2 $ where

\begin{equation}
\mathbb{E} \big\{ | \hat{x}_n^{(g_m)} |^2 \big\} = \frac{1}{N} \sum_{k=0}^{N-1} \big[ \textbf{W}_k^{(g)} \big]_{(:,m)}^H \textbf{R}_{\tilde{\textbf{y}}_{f,k}^{(g)}} \big[ \textbf{W}_k^{(g)} \big]_{(:,m)},\label{est_pow}
\end{equation}

\noindent
which can be obtained by using Parseval's relation for \eqref{X_k_g} and \eqref{x_n_g}. Covariance matrix of $ \tilde{\textbf{y}}_{f,k}^{(g)} $ given in \eqref{R_Y_k} is used for the calculation above. Since $ a^{(g_m)} $ and $ b^{(g_m)} $ do not depend on the time instance $ n $, they can be calculated once for given instantaneous effective channel matrices.

\subsection{Beamspace-aware Channel Estimation}\label{sec:ChannelEstimation}

In this paper, channel estimation is realized in reduced dimension. Instantaneous intra-group effective channel matrices $ \big\{ \textbf{H}_{eff,l}^{(g,g)} \big\}_{l=0}^{L-1} $ of each group are estimated during the uplink channel estimation phase in TDD mode where distinct pilot sequences with length $ T $ are assigned to each user in a group. Equivalent signal vector of group $ g $ at the BS at the end of the pilot sequence transmission is defined as

\begin{equation}
\begin{aligned}
\bar{\textbf{y}}^{(g)} & \triangleq \Big[ \big[ \tilde{\textbf{y}}_{0}^{(g)} \big]^H, \big[ \tilde{\textbf{y}}_{1}^{(g)} \big]^H, \cdots, \big[ \tilde{\textbf{y}}_{T-1}^{(g)} \big]^H  \Big]^H  \\
& = \sum_{g'=1}^{G} \Big( \textbf{X}^{(g')} \otimes \textbf{I}_{D_{g}} \Big) \bar{\textbf{h}}_{eff}^{(g,g')} + \Big( \textbf{I}_{T} \otimes \big[ \textbf{S^{(g)}} \big]^H \Big)\ \bar{\textbf{n}},
\end{aligned}\label{y_bar_g}
\end{equation}

\noindent
where $ \textbf{X}^{(g)} \in \mathbb{C}^{T \times K_g L} $ is the equivalent pilot matrix of group $ g $ and $ \bar{\textbf{h}}_{eff}^{(g,g')} \in \mathbb{C}^{D_g K_{g'} L \times 1} $ is the overall effective channel vector of group $ g' $ after the analog beamformer of group $ g $. Note that equivalent signal vector $ \bar{\textbf{y}}^{(g)} \in \mathbb{C}^{TD_g \times 1} $ is obtained by vertically concatenating received signal vectors after the analog beamformer $ \big\{ \tilde{\textbf{y}}_{n}^{(g)} \big\}_{n=0}^{T-1} $ in \eqref{y_n_g}. Similarly, equivalent noise vector $ \bar{\textbf{n}} \in \mathbb{C}^{TM \times 1} $ is constructed by vertically concatenating noise vectors $ \{ \textbf{n}_n \}_{n=0}^{T-1} $. Equivalent pilot matrix is defined as

\begin{equation}
\textbf{X}^{(g)} \triangleq \Big[ \textbf{X}^{(g_1)}, \textbf{X}^{(g_2)}, \cdots, \textbf{X}^{(g_{K_g})} \Big], \label{X}
\end{equation}

\noindent
where $ \big[ \textbf{X}^{(g_m)} \big]_{(i,j)} \triangleq x_{i-j}^{(g_m)} $, for $ i = 0,1,\dots,T-1 $ and $ j = 0,1,\dots,L-1 $. Pilot sequence entries of user $ m $ in group $ g $ are denoted by $ x_{n}^{(g_m)} $. Note that entries with negative indices can be taken as cyclic prefix or previously decoded symbols. Overall effective channel vector of group $ g' $ for the analog beamformer of group $ g $ is defined as

\begin{equation}
\bar{\textbf{h}}_{eff}^{(g,g')} \triangleq \Big[ \big[ \bar{\textbf{h}}_{eff}^{(g,g'_1)} \big]^H, \big[ \bar{\textbf{h}}_{eff}^{(g,g'_2)} \big]^H, \cdots, \big[ \bar{\textbf{h}}_{eff}^{(g,g'_{K{g'}})} \big]^H \Big]^H, \label{h_bar_g}
\end{equation}

\begin{equation}
\bar{\textbf{h}}_{eff}^{(g,g'_m)} \triangleq \Big[ \big[ \textbf{h}_{eff,0}^{(g,g'_m)} \big]^H, \big[ \textbf{h}_{eff,1}^{(g,g'_m)} \big]^H, \cdots, \big[ \textbf{h}_{eff,L-1}^{(g,g'_m)} \big]^H \Big]^H, \label{h_bar_g_m}
\end{equation}

\noindent
where $ \textbf{h}_{eff,l}^{(g,g'_m)} \triangleq \big[ \textbf{S}^{(g)} \big]^H \textbf{h}_l^{(g'_m)} $, or equivalently $ \textbf{h}_{eff,l}^{(g,g'_m)} \triangleq \big[ \textbf{H}_{eff,l}^{(g,g')} \big]_{(:,m)}, $ for $ l = 0,1,\dots,L-1 $. It is important to note that we are only interested in intra-group channel vectors. In other words, our aim is to estimate $ \bar{\textbf{h}}_{eff}^{(g,g)} $ from \eqref{y_bar_g} as

\begin{equation}
\hat{\bar{\textbf{h}}}_{eff}^{(g,g)} = \big[ \textbf{Z}^{(g)} \big]^H \bar{\textbf{y}}^{(g)} \label{h_eff_est}
\end{equation}

\noindent
where $ \textbf{Z}^{(g)} \in \textbf{C}^{D_g K_g L \times T D_g} $ is the channel estimator matrix. In this paper, LMMSE and LS channel estimators are considered. Expressions for these estimators can be found in Appendix \ref{channel}. In order to show the accuracy of channel estimation, we calculate nMSE which is defined as

\begin{equation}
\begin{aligned}
& nMSE^{(g)} \triangleq \frac{ \mathbb{E} \bigg\{ \norm{ \bar{\textbf{h}}_{eff}^{(g,g)} - \hat{\bar{\textbf{h}}}_{eff}^{(g,g)} }^2 \bigg\} }{ \mathbb{E} \bigg\{ \norm{ \bar{\textbf{h}}_{eff}^{(g,g)} }^2 \bigg\} }  \\
& \! = \frac{ \mathrm{Tr} \Big\{ \textbf{R}_{ \bar{\textbf{h}}_{eff}^{(g,g)}} \Big\} \! + \! \mathrm{Tr} \Big\{ \textbf{R}_{ \hat{\bar{\textbf{h}}}_{eff}^{(g,g)}} \Big\} \! - \! 2 \mathrm{Re} \Big\{ \mathrm{Tr} \Big\{ \textbf{R}_{ \hat{\bar{\textbf{h}}}_{eff}^{(g,g)} \bar{\textbf{h}}_{eff}^{(g,g)}} \Big\} \Big\} }{ \mathrm{Tr} \Big\{ \textbf{R}_{ \bar{\textbf{h}}_{eff}^{(g,g)}} \Big\} }
\end{aligned} \label{nMSE}
\end{equation}

\noindent
where covariance matrix of $ \hat{\bar{\textbf{h}}}_{eff}^{(g,g)} $ and covariance between $ \hat{\bar{\textbf{h}}}_{eff}^{(g,g)} $ and $ \bar{\textbf{h}}_{eff}^{(g,g)} $ can be computed as

\begin{equation}
\textbf{R}_{ \hat{\bar{\textbf{h}}}_{eff}^{(g,g)} } = \mathbb{E} \Big\{ \hat{\bar{\textbf{h}}}_{eff}^{(g,g)} \big[ \hat{\bar{\textbf{h}}}_{eff}^{(g,g)} \big]^H \Big\} = \big[ \textbf{Z}^{(g)} \big]^H \textbf{R}_{\bar{\textbf{y}}^{(g)}} \textbf{Z}^{(g)}, \label{R_h_eff_cap}
\end{equation}

\begin{equation}
\textbf{R}_{ \hat{\bar{\textbf{h}}}_{eff}^{(g,g)} \bar{\textbf{h}}_{eff}^{(g,g)}} = \mathbb{E} \Big\{ \hat{\bar{\textbf{h}}}_{eff}^{(g,g)} \big[ \bar{\textbf{h}}_{eff}^{(g,g)} \big]^H \Big\} = \big[ \textbf{Z}^{(g)} \big]^H \textbf{R}_{\bar{\textbf{y}}^{(g)} \bar{\textbf{h}}_{eff}^{(g,g)}}. \label{R_h_eff_h_eff_cap}
\end{equation}

Expressions of $ \textbf{R}_{\bar{\textbf{y}}^{(g)} \bar{\textbf{h}}_{eff}^{(g,g)}} $, $ \textbf{R}_{\bar{\textbf{y}}^{(g)}} $ and $ \textbf{R}_{\bar{\textbf{h}}_{eff}^{(g,g)}} $ are given in \eqref{R_y_g_h_eff}, \eqref{R_y_g_bar} and \eqref{R_h_eff} in Appendix \ref{channel}, respectively. There is a final remark for the described channel estimation scheme.

\begin{rem}
With the assumption that analog beamformer suppresses inter-group interference perfectly, training periods of groups does not have to be synchronized. In this case, training symbols of other groups ($ \textbf{X}^{(g')}, g' \neq g $) can be assumed to be random. Furthermore, users in a group should use different training sequences whereas users in other groups can reuse the same sequences with properly designed analog beamformers. Hence, it can be inferred that this channel estimation scheme is robust against pilot contamination. \label{rem3}
\end{rem}

\vspace{-6mm}
\section{Numerical Results}

In this section, proposed constrained analog beamformers are compared by using the performance measures in Section \ref{perf}. BS has an ULA with $ M = 128 $ antennas whereas users have single-antenna elements. Carrier frequency is selected as $30$GHz while the antenna spacing of the ULA is set to half the wavelength. A predetermined user-grouping is employed with $ G = 4 $ groups where each group has $ K_g = 2 $ users\footnote{Here, we assume that users come in groups, either by nature or by the use of a proper user grouping algorithm \cite{JSDM_Group,Joint_Group}.}. Although total number of delays is $ L = 32 $, users do not have active MPCs at every delay. In the simulations, angular spread $ \Delta_l^{(g_m)} $ of all users is set to $2^\circ$ for every MPC. Power profile function $ \rho_l^{(g_m)}(\theta) $ is assumed to be uniform in the angular spread of the MPC and channel gain $ \sqrt{\gamma^{(g_m)}} $ is set to $1$ for all users. MPC cluster powers of users are assumed to be the same. Noise variance $ N_0 $ is taken as $1$ for simplicity. 

The scenario used for simulations is given in Table \ref{table1}. This table shows the active MPC indices of groups and mean AoAs of users related to each active MPC. Group-1 is selected as the mobile group with shifting angle of $ \phi $ which is swept from $-45^\circ $ to $45^\circ $ with increments of $0.1^\circ$ for simulation purposes. It is important to note that CCMs are slowly varying parameters. However, the reason to use a mobile group is to obtain results for different angular profiles and observe the effect of overlapping MPCs. Hence, we are only interested in the results of the mobile group which is Group-1. For each $ \phi $, CCMs and analog beamformer of Group-1 are calculated, then ergodic capacity and nMSE are obtained. Ergodic capacity of user $ m $ in Group-1 and nMSE of Group-1 corresponding to shifting angle $ \phi $ are denoted by $ C^{(1_m)}_\phi $ and $ nMSE^{(1)}_\phi $, respectively. Ergodic capacity results are obtained with perfect CSI knowledge.

\begin{table}[!h]
\small
\caption{Angle-delay profile of groups}
\begin{center}
\vspace{-3mm}
\begin{tabular}{|c|c|c|}
\hline
Group      & MPC index  & Mean AoAs $\big( \mu_l^{(g_m)} \big) $ of users\\ \hline
\multirow{3}{*}{{1 (Mobile)}} & 0         & {\{}$\phi$-15.5$^\circ$, $\phi$-14.5$^\circ${\}}      \\ \cline{2-3} 
                   & 5         & {\{}$\phi$-2.5$^\circ$, $\phi$-1.5$^\circ${\}}      \\ \cline{2-3} 
                   & 11        & {\{}$\phi$+16.5$^\circ$, $\phi$+17.5$^\circ${\}}      \\ \hline
\multirow{2}{*}{2} & 3         & {\{}40.5$^\circ$, 41.5$^\circ${\}}     \\ \cline{2-3} 
                   & 9         & {\{}20.5$^\circ$, 21.5$^\circ${\}}      \\ \hline
\multirow{2}{*}{3} & 8         & {\{}-10.5$^\circ$, -9.5$^\circ${\}}      \\ \cline{2-3} 
                   & 17        & {\{}-20.5$^\circ$, -19.5$^\circ${\}}    \\ \hline
\multirow{1}{*}{4} & 29        & {\{}-40.5$^\circ$, -39.5$^\circ${\}}       \\ \hline
\end{tabular}
\label{table1}
\end{center}
\end{table}

Angle-delay map of groups is given in Fig. \ref{PDP} in order to visualize the scenario given in Table \ref{table1}. In this figure, $ \phi $ of the mobile group is set to $15^\circ$. This is a typical mm-wave scenario where we observe a sparse angle-delay profile.

\begin{figure}[!h]
\centering
\includegraphics[width=0.44\textwidth]{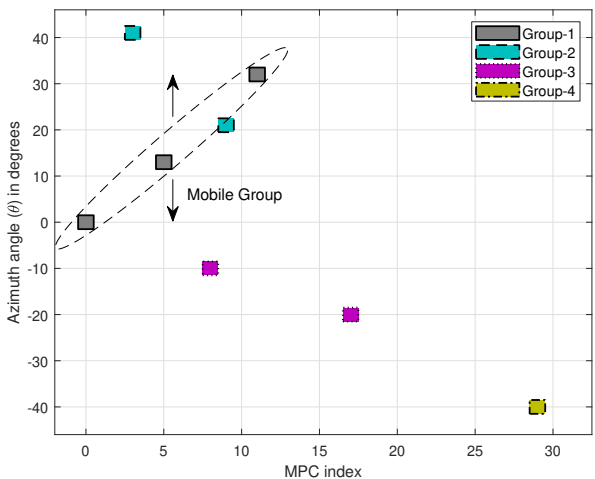}
\caption{Angle-delay map of groups with $ \phi = 15^\circ $ }
\label{PDP}
\vspace{-3mm}
\end{figure}

\subsection{JSDM with Fully Connected Arrays}

In this section, constrained fully connected array structures are compared. Unconstrained near-optimal analog beamformer GEB introduced in Section \ref{sec:analog} is used as a benchmark. DFT beamformer, phase extraction and PE-AM in Algorithm \ref{PE-AM} are the considered constrained beamformers. 

Firstly, we will introduce beampattern metric which shows the power that can be attained at the angle of interest $ \theta $ for a given analog beamformer. Interference suppression capability of beamformers can be observed by using this metric. Beampattern as a function of $ \theta $ for given analog beamformer $ \textbf{S}^{(g)} $ can be defined as

\begin{equation}
B(\theta) = \textbf{u}(\theta)^H \textbf{S}^{(g)} \Big( \big[ \textbf{S}^{(g)} \big]^H \textbf{S}^{(g)} \Big)^{-1} \big[ \textbf{S}^{(g)} \big]^H \textbf{u}(\theta) . \label{beam}
\end{equation}

Fig. \ref{beampattern_full} shows the beampatterns of Group-1 for PE-AM and DFT beamformer with $ \phi = 10^\circ $ and $ D_1 = 4 $ where symbol energy of each group is set to $ 40 $ dB. It is seen that both beamformers can form narrow beams. Although PE-AM does not provide deep nulls, we observed that it can provide moderate suppression for MPCs of other groups whereas DFT beamformer does not consider these MPCs. If beampatterns of PE-AM and DFT beamformer are compared, this suppression appears more significant for interfering angular regions that are close to the angular region of intended group.

\vspace{-2mm}
\begin{figure}[!h]
\centering
\subfloat[Subfigure 1 list of figures text][PE-AM]{
\includegraphics[width=0.46\textwidth]{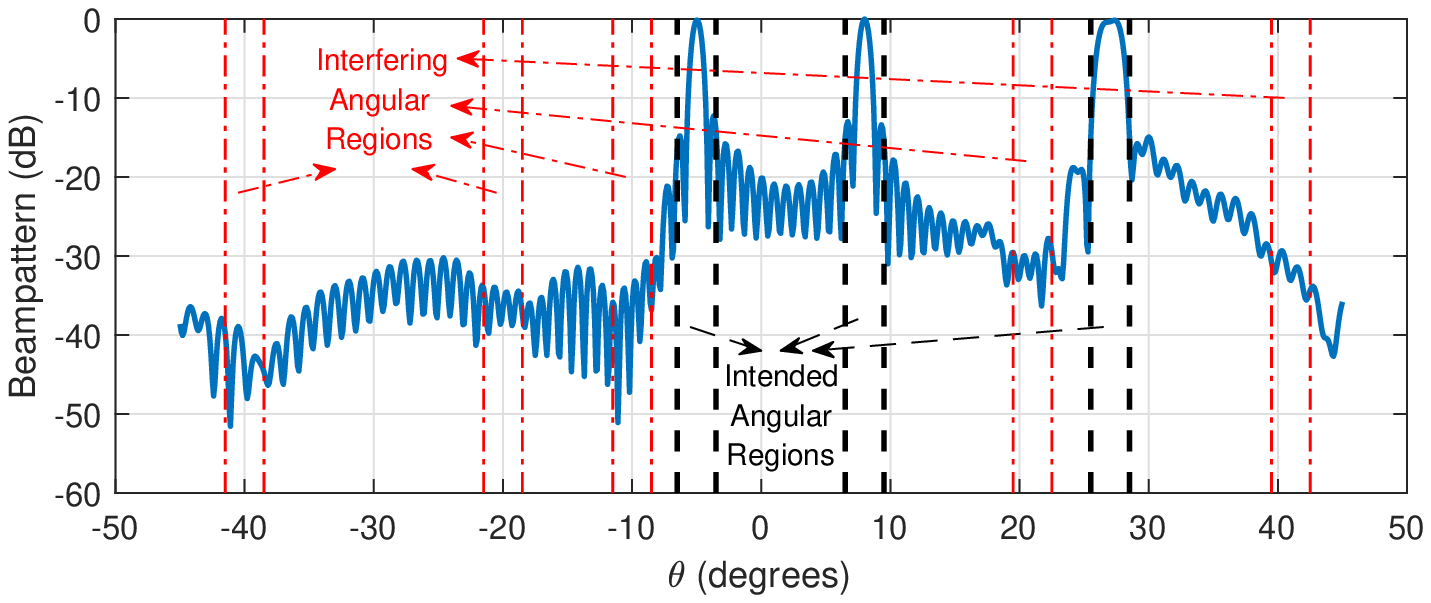}
\label{beampattern_full_D_4_2}}\hfill
\vspace{1mm}
\subfloat[Subfigure 2 list of figures text][DFT beamformer]{
\includegraphics[width=0.46\textwidth]{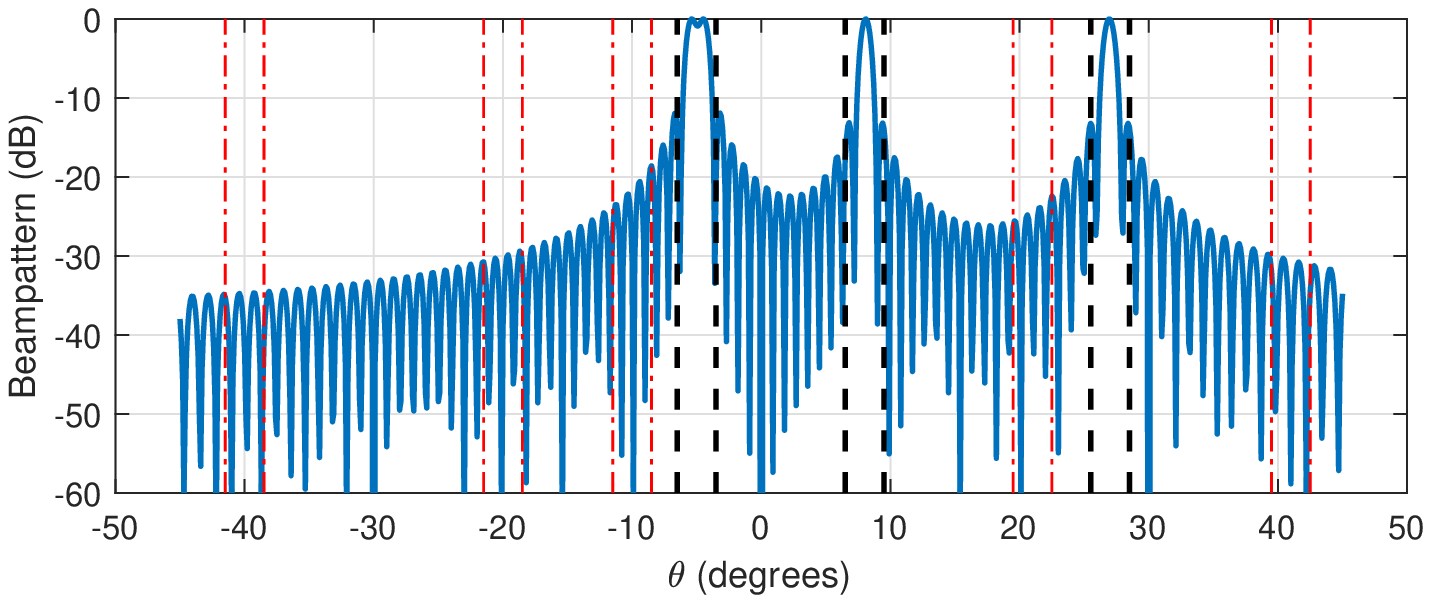}
\label{beampattern_full_D_4_4}}

\caption{Beampatterns of Group-1 for fully connected array structures with $ \phi = 10^\circ $, $ D_1 = 4 $ and $ E_s^{(g)} = 40 $ dB for $ \forall g $}
\label{beampattern_full}
\vspace{-1mm}
\end{figure}

\begin{figure*}[!h]
\centering

\subfloat[Subfigure 1 list of figures text][$ D_1 = 2 $]{
\includegraphics[width=0.4\textwidth]{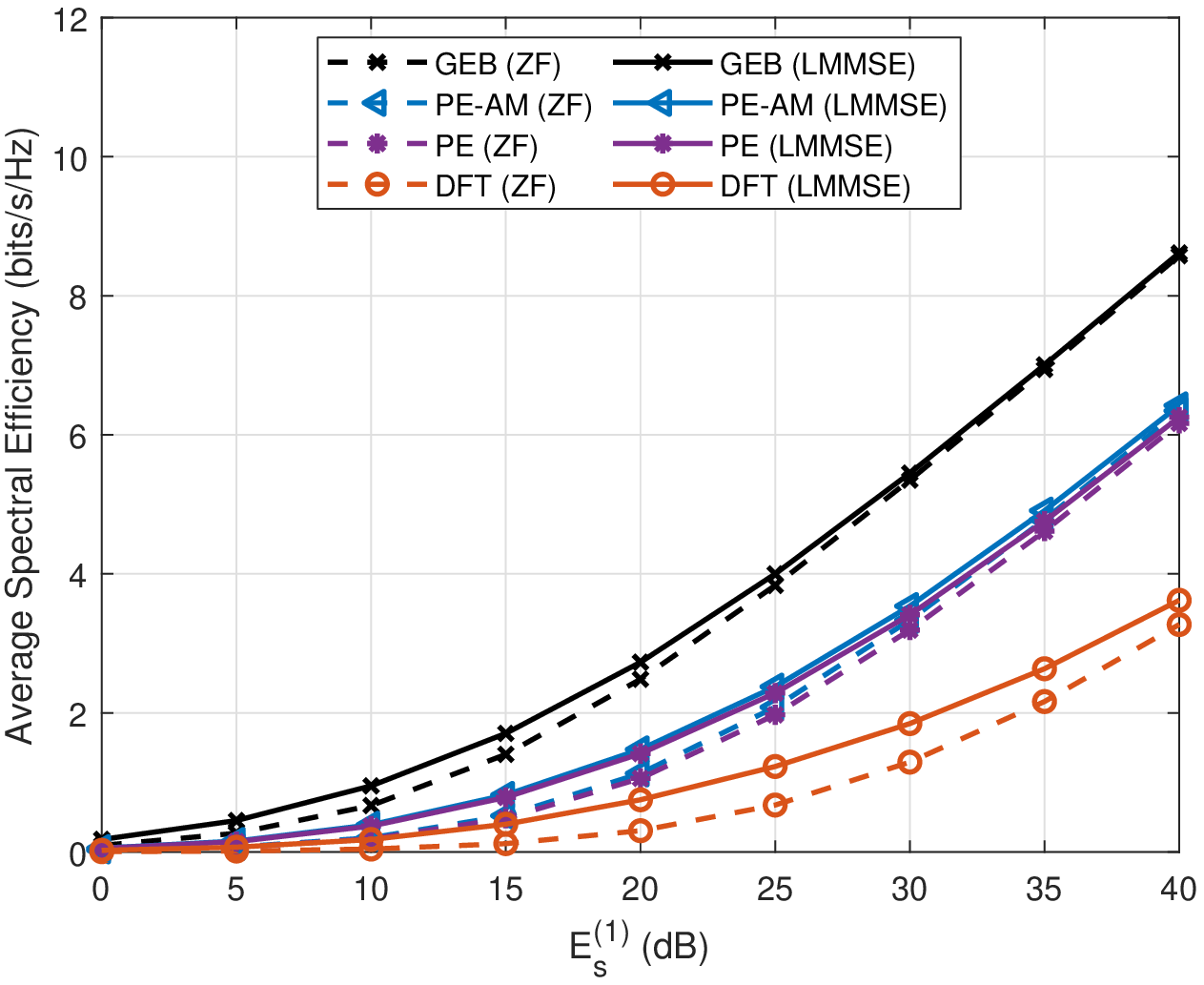}
\label{Rate_full_D_2_3}}\hspace{0.05\textwidth}
\subfloat[Subfigure 2 list of figures text][$ D_1 = 4 $]{
\includegraphics[width=0.4\textwidth]{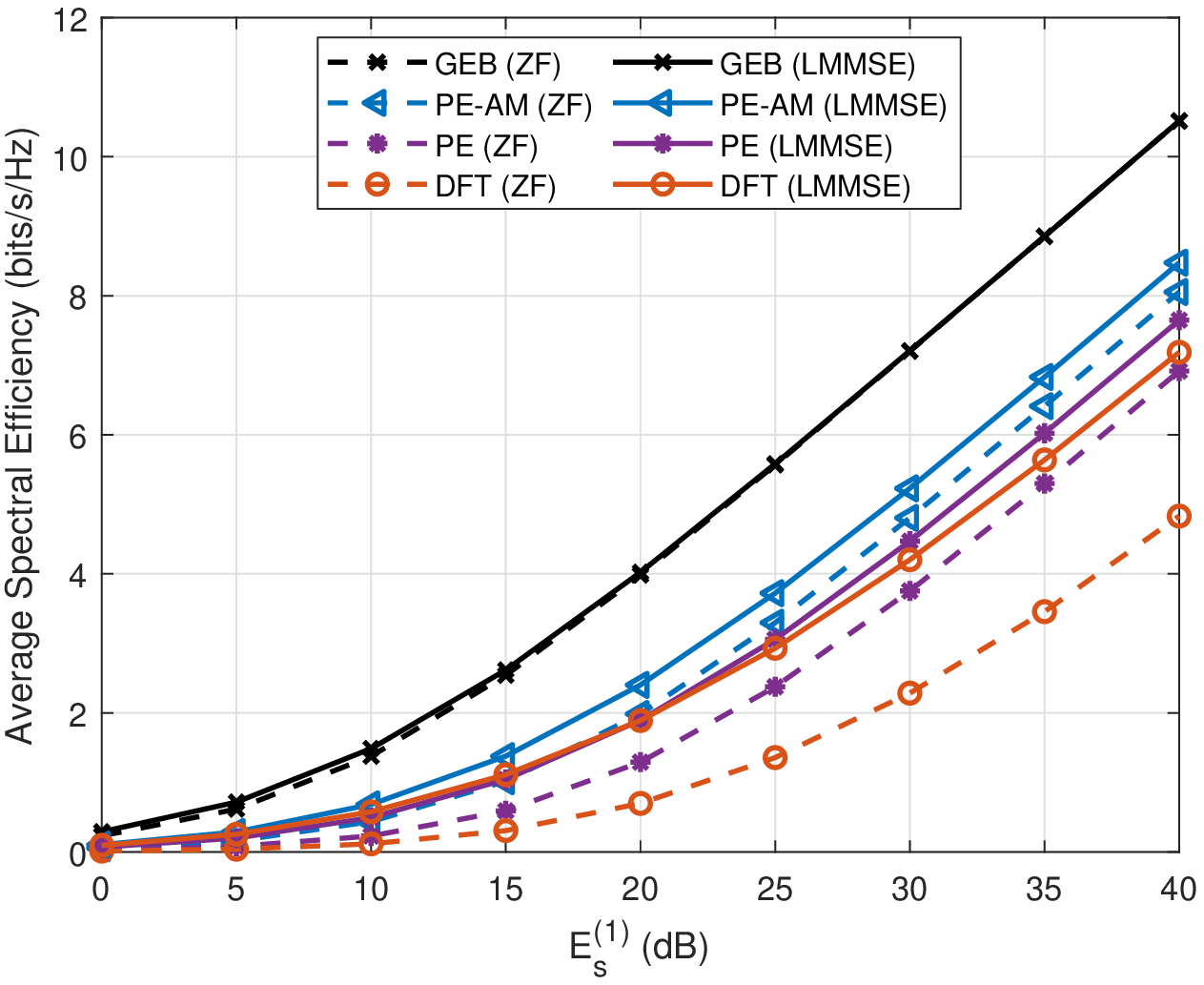}
\label{Rate_full_D_4_3}}

\caption{Average spectral efficiency of Group-1 vs. $ E_s^{(1)} $ for fully connected array structures with $ E_s^{(g')} = 40 $ dB for $ g' \neq 1 $}
\label{Rate_full}

\end{figure*}

\begin{figure*}[!h]
\centering

\subfloat[Subfigure 1 list of figures text][$ D_1 = 2 $]{
\includegraphics[width=0.4\textwidth]{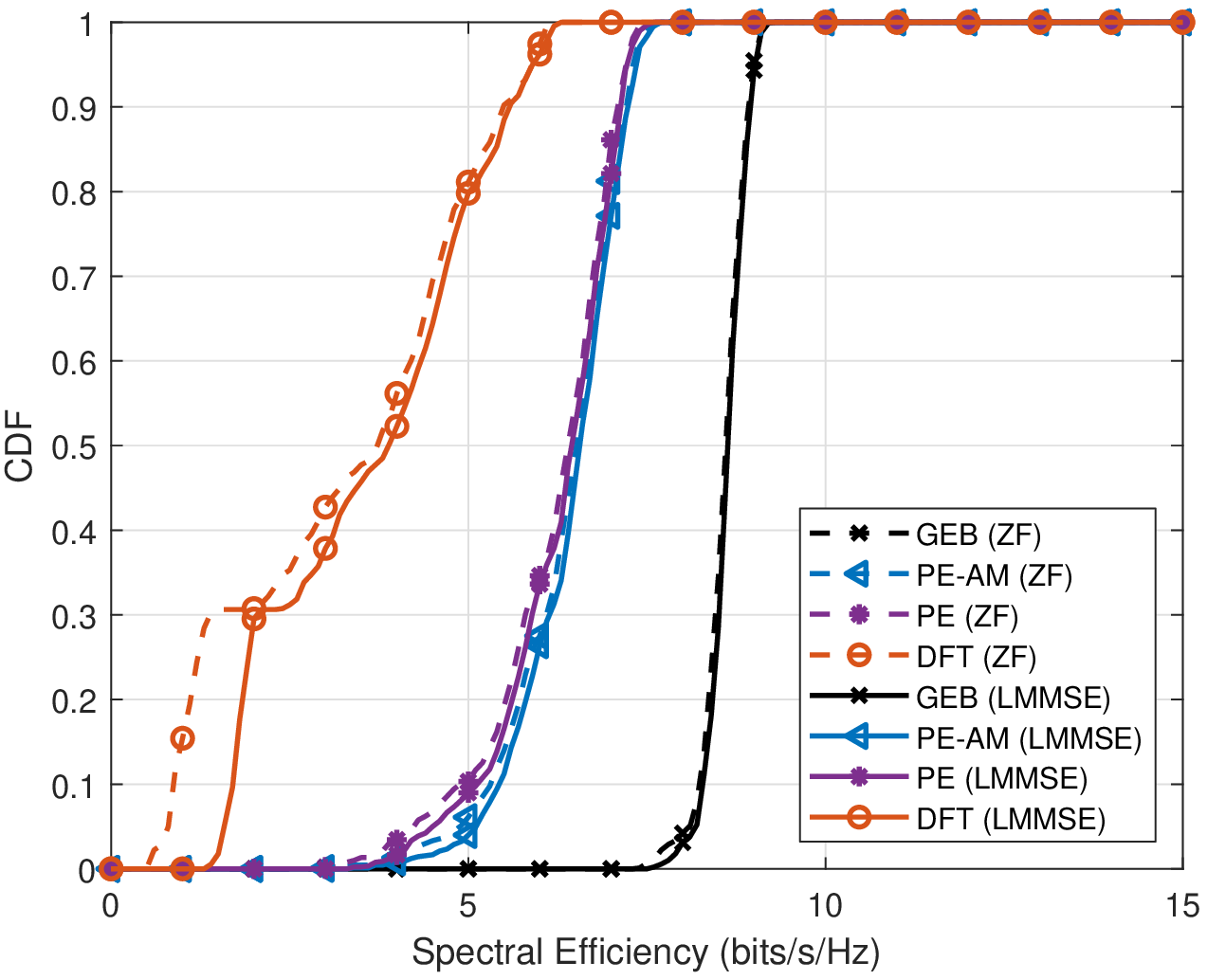}
\label{CDF_full_D_4_1}}\hspace{0.05\textwidth}
\subfloat[Subfigure 2 list of figures text][$ D_1 = 4 $]{
\includegraphics[width=0.4\textwidth]{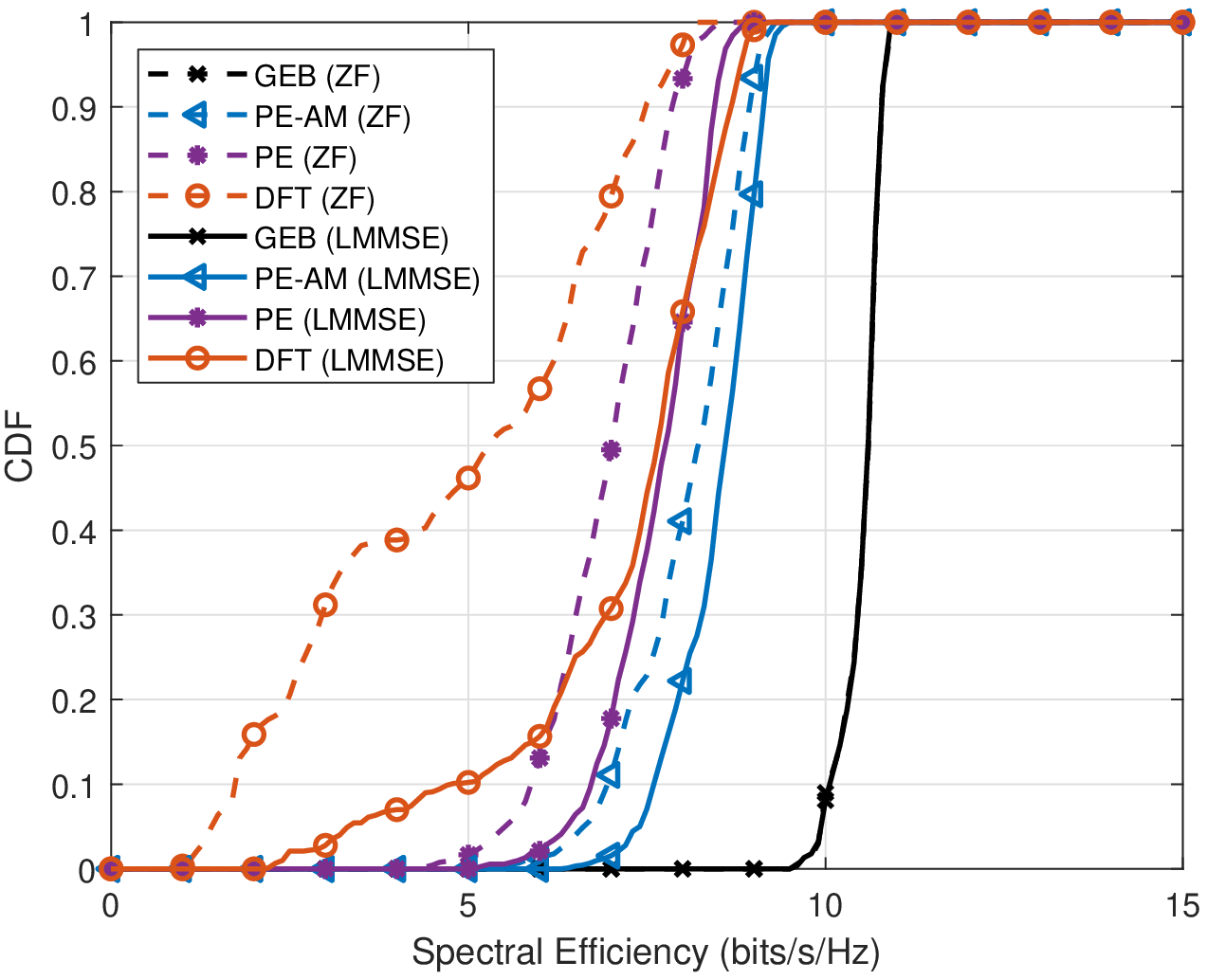}
\label{CDF_full_D_4_3}}

\caption{CDF of spectral efficiency of Group-1 for fully connected array structures with $ E_s^{(g)} = 40 $ dB for $ \forall g $}
\label{CDF_full}
\vspace{-3mm}
\end{figure*}

Fig. \ref{Rate_full} shows the average spectral efficiency of Group-1 with changing symbol energy where symbol energy of other groups is set to $ E_s^{(g')} = 40 $ dB for $ g' \neq 1 $. Note that average spectral efficiency is found by taking the average of ergodic capacity over both users and shifting angle $ \phi $, i.e., $ C^{(1)}_{avg} = \mathbb{E}_{\phi, m} \Big\{ C^{(1_m)}_{\phi} \Big\} $. It can be observed that LMMSE type digital beamformer does not provide a significant improvement over ZF for analog beamformers except DFT beamformer with $ D_1 = 2 $. However, LMMSE provides improvement for all constrained analog beamformers, especially for DFT beamformer, with $ D_1 = 4 $ as it can efficiently suppress residual inter-group interference when the number of RF chains is greater than the number of users. It should be noted that LMMSE and ZF type digital beamformers perform the same for GEB at high symbol energy which shows that there is no residual interference. Furthermore, it is seen that there is a significant gap between DFT beamformer and other constrained beamformers for $ D_1 = 2 $ whereas the gap between PE and PE-AM is very small. However, there is a gap between PE and PE-AM for $ D_1 = 4 $ which is expected since increase in RF chains corresponds to increase in the dimension of the compensation matrix for PE-AM. LMMSE type digital beamformer makes the DFT beamformer perform almost as good as PE for $ D_1 = 4 $.

Previous results show the average performance of analog beamformers in the mobile group. However, angular regions of MPCs belonging to Group-1 overlaps with angular regions of MPCs of other groups at certain $ \phi $ values. It is important to observe the degradation in spectral efficiency with different constrained analog beamformers for overlapping situations. In other words, we need to observe the variation in spectral efficiency to obtain the outage capacity. Let the cumulative density function (CDF) of spectral efficiency of Group-1 is defined as $ P \big( C^{(g)}_\phi<c \big) $ where average spectral efficiency of Group-1 for shifting angle $ \phi $ is denoted by $ C^{(1)}_\phi = \mathbb{E}_m \Big\{ C^{(1_m)}_\phi \Big\} $. Then, CDF curves correspond to outage probabilities for given spectral efficiency values. In Fig. \ref{CDF_full}, CDF curves with different analog beamformers are compared for $ D_1 = 2 $ and $ D_1 = 4 $ where symbol energy of each group is set to $ 40 $ dB. It is observed that CDF curves of DFT beamformer is wider compared to others which results in higher outage probability. For example, outage probability of PE-AM with $ D_1 = 4 $ is close to zero at 7 bits/s/Hz while it is above 0.3 for DFT beamformer. On the other hand, GEB has sharp CDF curves for both $ D_1 = 2 $ and $ D_1 = 4 $ cases which means that it is a robust analog beamforming method. Furthermore, it is observed that PE and PE-AM has fairly sharp CDF curves especially for $ D_1 = 4 $ case.

\begin{figure*}[!h]
\centering

\subfloat[Subfigure 1 list of figures text][Average nMSE vs. $ T $ with $ E_s^{(1)} = 30 $ dB]{
\includegraphics[width=0.4\textwidth]{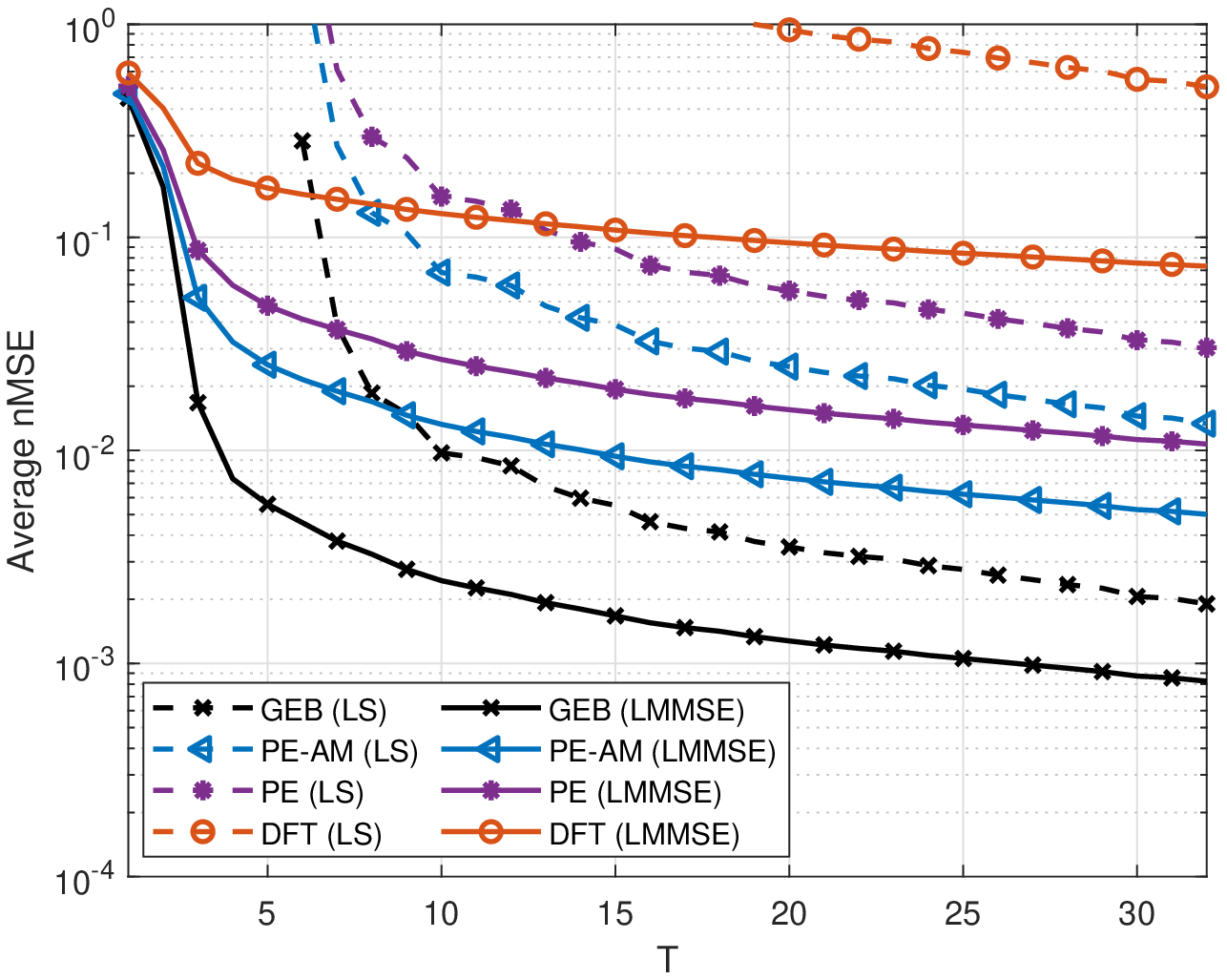}
\label{nMSE_T_full_D_4_2}}\hspace{0.05\textwidth}
\subfloat[Subfigure 2 list of figures text][Average nMSE vs. $ E_s^{(1)} $ with $ T = 10 $]{
\includegraphics[width=0.4\textwidth]{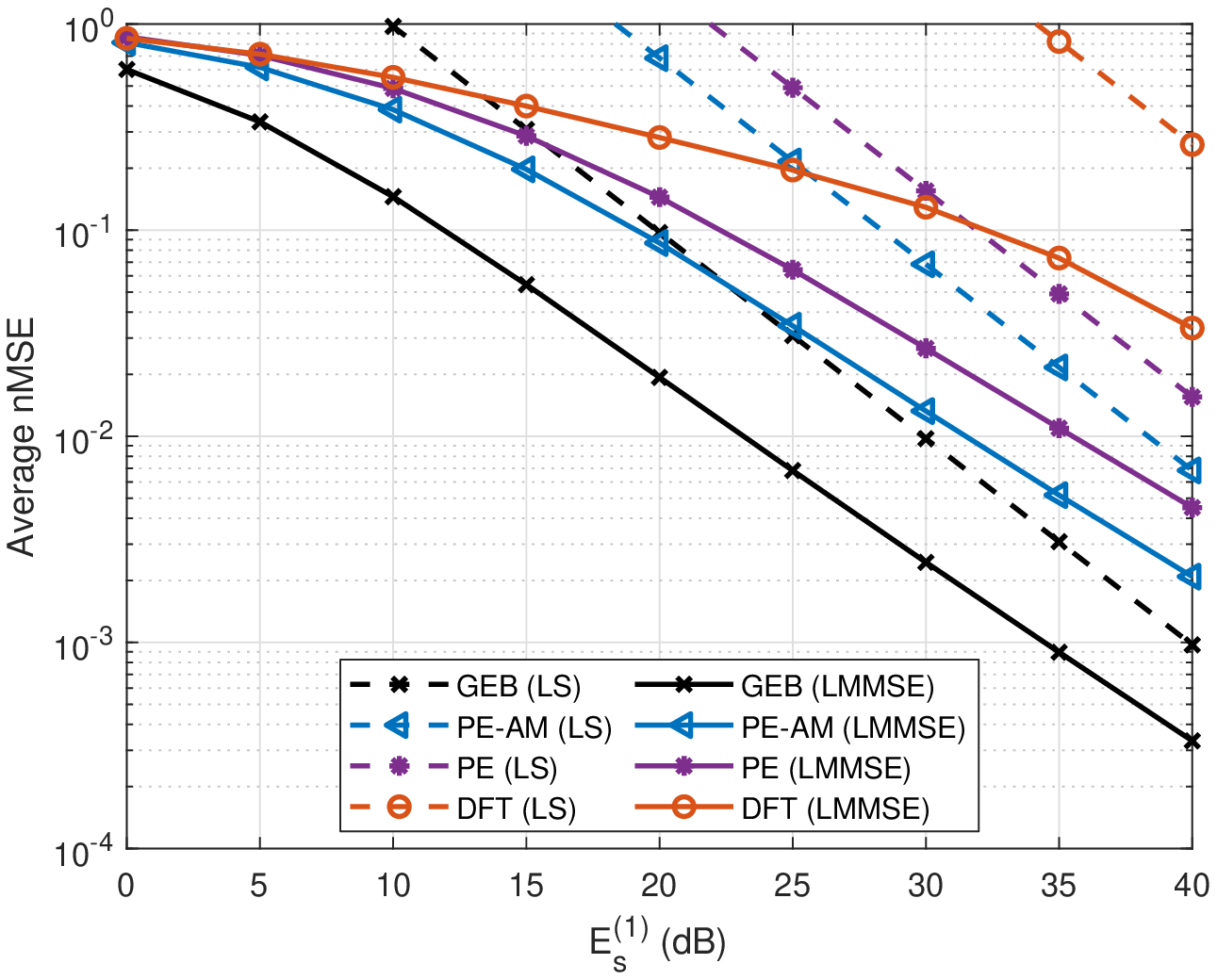}
\label{nMSE_SNR_full_D_4_2}}

\caption{Average nMSE of Group-1 for fully connected array structures with $ D_1 = 4 $ and $ E_s^{(g')} = 40 $ dB for $ g' \neq 1 $}
\label{nMSE_full}
\vspace{-3mm}
\end{figure*}

Channel estimation accuracy is another performance measure that is considered in this paper. Fig. \ref{nMSE_full} shows the average nMSE of Group-1 with $ D_1 = 4 $ where symbol energy of other groups is set to $ E_s^{(g')} = 40 $ dB for $ g' \neq 1 $. Average of nMSE is taken over shifting angle $ \phi $, i.e., $ nMSE^{(1)}_{avg} = \mathbb{E}_{\phi} \Big\{ nMSE^{(1)}_{\phi} \Big\} $. It can be observed that LMMSE estimator outperforms LS estimator especially with short pilot lengths for all analog beamformers. Moreover, it is important to note that nMSE of DFT beamformer is much higher compared to others even with LMMSE estimator while nMSE of PE-AM is less than nMSE of both PE and DFT beamformer.

In spectral efficiency analysis, DFT beamformer performed close to PE for $ D_1 = 4 $ with perfect CSI. However, nMSE results show that spectral efficiency performance of DFT would degrade more compared to other ones if the channel estimation errors are considered. Considering the beampattern, outage capacity and nMSE results, it can be concluded that DFT beamformer should be avoided for JSDM framework especially when the near-far effect is observed  as it does not consider MPCs other groups\footnote{The near-far effect stems from the fact that the average received signal strength of different UEs may differ significantly depending on their distance to the BS.}. Furthermore, PE-AM yields better performance than PE especially when the number of RF chains is not highly restricted since dimensions of compensation matrix are determined by this number. Consequently, PE-AM should be preferred to obtain fully connected constrained analog beamformers.

\subsection{JSDM with Partially Connected Arrays}

In this section, partially connected array structures are considered. While symbol energy of other groups (i.e., inter-group interference) was set to $ 40 $ dB for fully connected structures, we consider a moderate interference scenario for partially connected arrays where symbol energy of interfering groups is set to $ 20 $ dB since even constrained analog beamformers with fully connected arrays do not attain the performance of GEB at high interference. Furthermore, we merge Group-1 and Group-2 and denote it by Group-1$'$ which includes $ 4 $ users and $ 5 $ different MPC clusters. Relative angular regions of MPCs of the merged group is taken as in Fig. \ref{PDP} and the center of the MPCs is swept with shifting angle $ \phi $ from $-45^\circ $ to $45^\circ $ with increments of $0.1^\circ $ as in previous part. We want to serve more users with a single analog beamformer which is the reason for merging groups. We will only be interested in the merged group and design a partially connected array for the analog beamformer of this group\footnote{Other groups can be served by using another array at the BS or scheduling them in time and frequency resources.}. Interference originated from Group-3 and Group-4 can be considered as any interference with known covariance matrix. We will consider two different fixed partially connected arrays. The first one is called adjacent or ordered which has a connection matrix $ \boldsymbol{\Pi}_{ord}^{1'} = \textbf{I}_{D_{1'}} \otimes \textbf{1}_{M/D_{1'}} $, where $ \textbf{1}_{M/D_{1'}} \in \mathbb{Z}^{M/D_{1'}} $ is a vector of ones. Second type of fixed partially connected arrays is called interlaced which has a connection matrix $ \boldsymbol{\Pi}_{int}^{1'} = \textbf{1}_{M/D_{1'}} \otimes \textbf{I}_{D_{1'}} $.

In Fig. \ref{beampattern_part}, beampatterns of Group-1$'$ with fixed and dynamic partially connected arrays are given for $ D_{1'} = 8 $ where MPCs are located as in Fig. \ref{PDP}. We selected the number of RF chains as 8 since it is the minimum power of two that is larger than the number of MPCs. Symbol energy of Group-1 and Group-2 before merging was set to $ 30 $ dB while symbol energy of other groups is set to $ 20 $ dB. We observe that ordered array forms wide beams whereas dynamic subarray can form narrow beams located in the angular regions of Group-1$'$. Hence, dynamic subarray preserved the multipath diversity even though beam power is slightly higher at interfering angular regions. In addition, interference suppression of ordered array would be lost due to wide beams when the MPCs are close to each other in the mobile scenario. On the other hand, interlaced array completely lost interference suppression as it suffers from grating lobe effect which occurs when the antenna spacing of an array is larger than half the wavelength. It is important to note that the proposed dynamic array is seen to effectively compromise between avoiding grating lobes and forming narrower beams.

\begin{figure}[!h]
\centering

\subfloat[Subfigure 1 list of figures text][Fixed - Ordered]{
\includegraphics[width=0.46\textwidth]{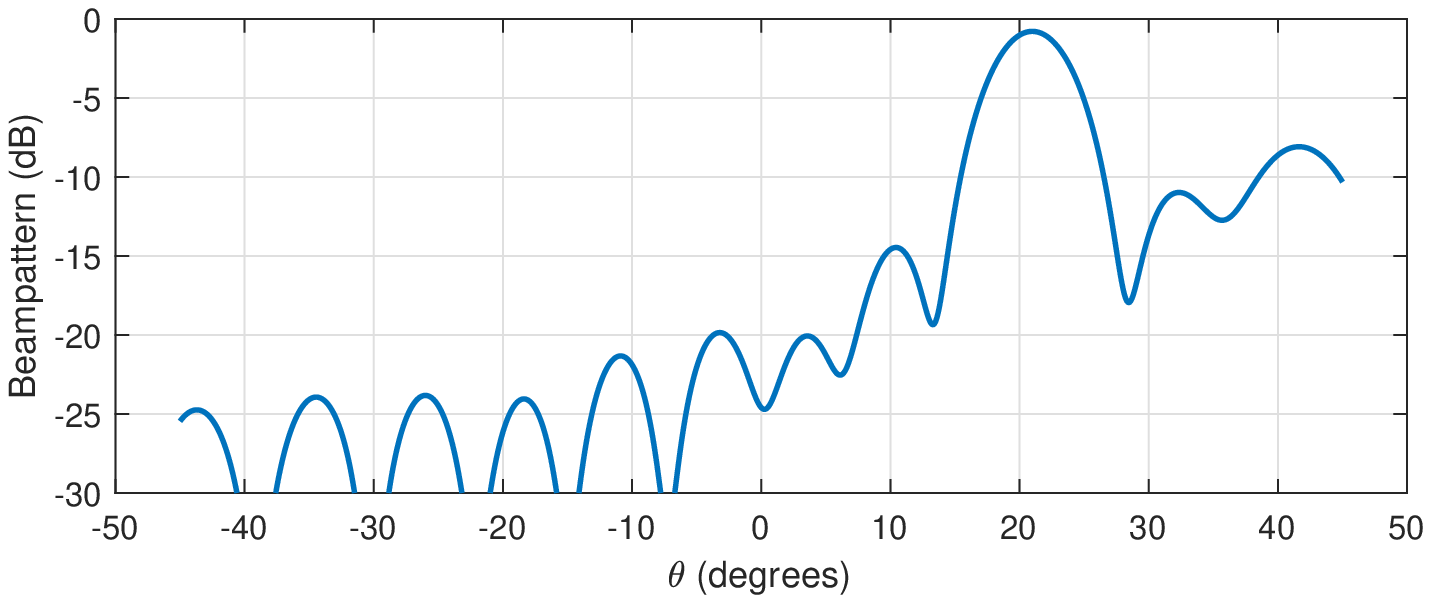}
\label{beampattern_part_D_8_5}}\hfill
\vspace{1mm}
\subfloat[Subfigure 2 list of figures text][Fixed - Interlaced]{
\includegraphics[width=0.46\textwidth]{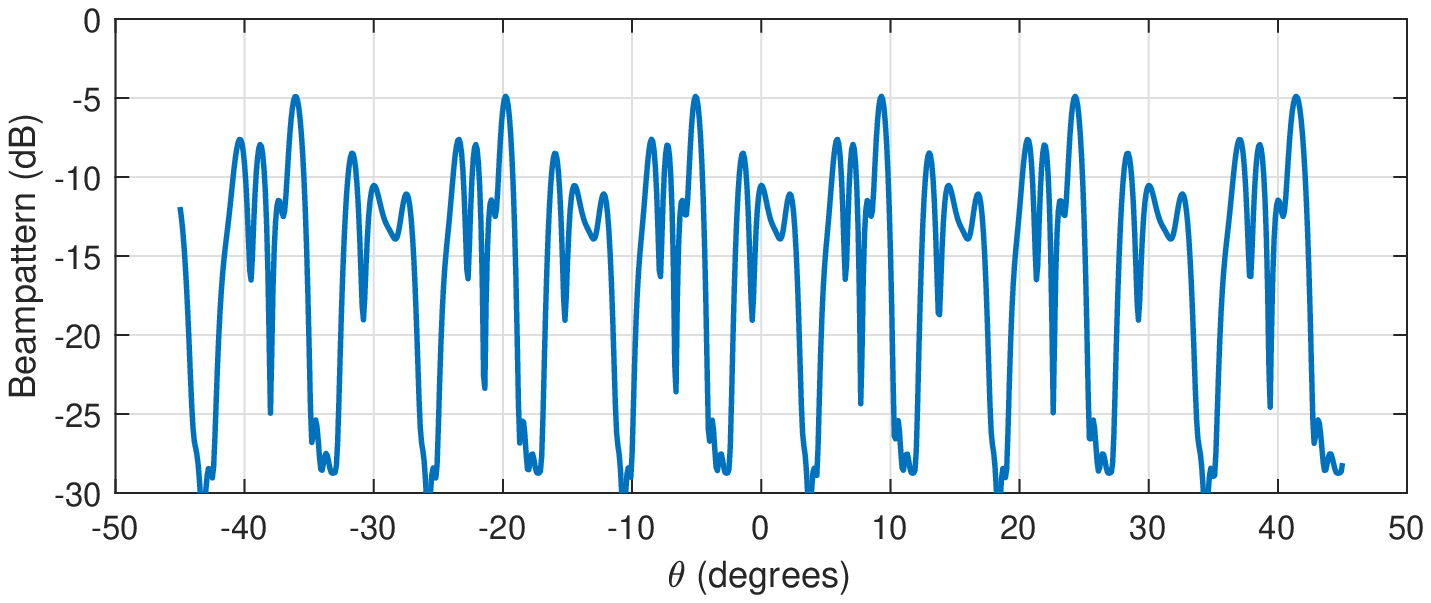}
\label{beampattern_part_D_8_6}}\hfill
\vspace{1mm}
\subfloat[Subfigure 2 list of figures text][Dynamic]{
\includegraphics[width=0.46\textwidth]{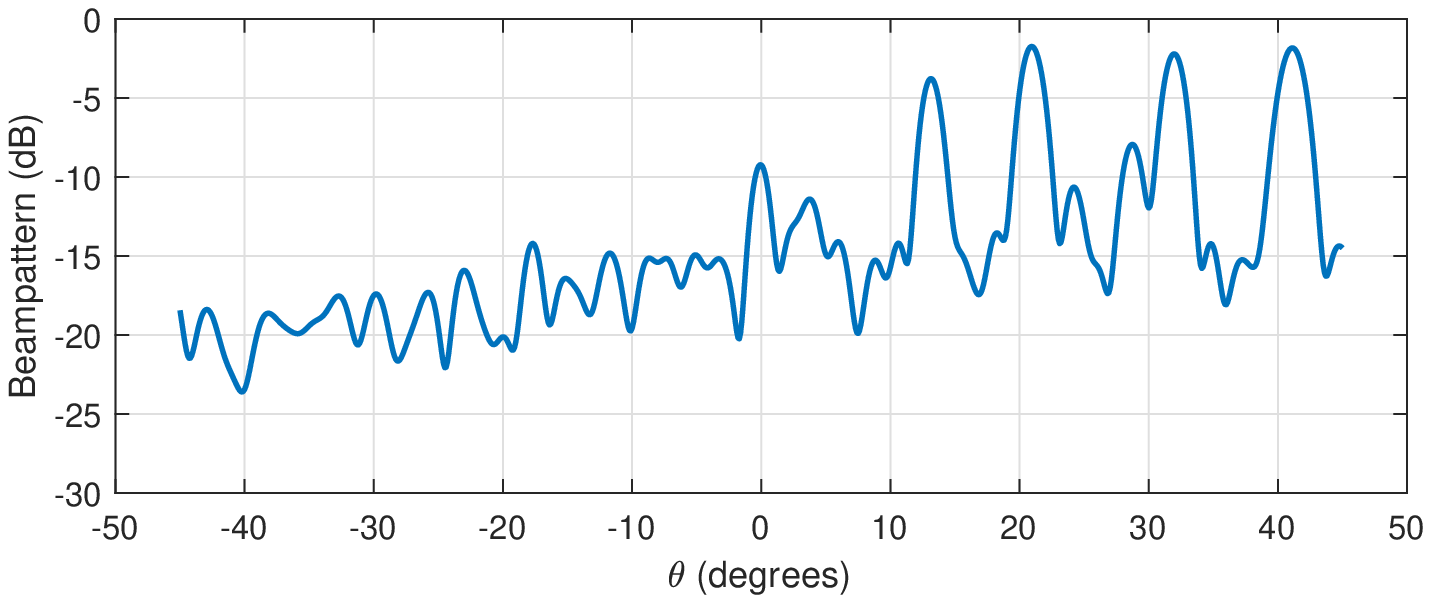}
\label{beampattern_part_D_8_7}}

\caption{Beampatterns of Group-1$'$ for partially connected array structures with $ D_1 = 8 $, $ E_s^{(1)} = E_s^{(2)} = 30 $ dB and $ E_s^{(3)} = E_s^{(4)} = 20 $ dB}
\label{beampattern_part}
\end{figure}

Average spectral efficiency with changing symbol energy of Group-1$'$ is given in Fig. \ref{Rate_part} for $ D_{1'} = 8 $ where fixed and dynamic subarrays and fully connected PE-AM and GEB are considered while symbol energy of other groups are set to $ 20 $ dB. Firstly, it is important to note that PE-AM performs very close to GEB for this scenario where moderate interference is considered. If we compare partially connected array designs, dynamic subarray outperforms both fixed array structures. Another important point is that the gap between LMMSE and ZF type digital beamformers is smaller for the dynamic subarray which indicates that residual interference after the analog beamformer is smaller compared to fixed subarrays. Moreover, the gap between dynamic and fixed subarrays is larger in this paper compared to prior work \cite{Dynamic,DynamicWideband2}. On the other hand, ordered array has higher average spectral efficiency than interlaced array as expected. We can make similar comments if we compare CDF of spectral efficiency curves which are given in Fig. \ref{CDF_part}. We observe that CDF curves of dynamic subarray are parallel to CDF curves of PE-AM and GEB whereas fixed arrays, especially interlaced structure, has wider CDF curves. This result shows the robustness of the dynamic subarray structure. In conclusion, PE-AM attains the performance of GEB and proposed dynamic subarray algorithm performs close to fully connected arrays which shows the resilience of proposed algorithms against interference for the moderate interference scenario.

\begin{figure}[!h]
\centering
\includegraphics[width=0.40\textwidth]{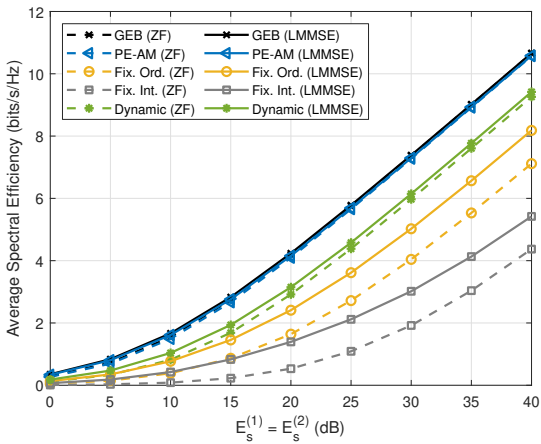}
\caption{Average spectral efficiency vs. $ E_s^{(1)} = E_s^{(2)} $ of Group-1$'$ for partially connected array structures with $ D_{1'} = 8 $ and $ E_s^{(3)} = E_s^{(4)} = 20 $ dB}
\label{Rate_part}
\end{figure}

\begin{figure}[!h]
\centering
\includegraphics[width=0.40\textwidth]{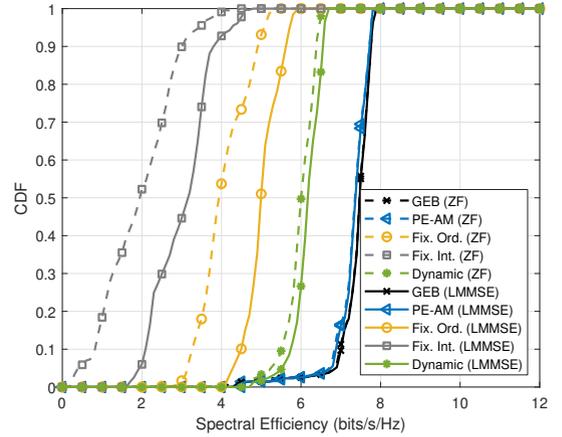}
\caption{CDF of spectral efficiency of Group-1$'$ for partially connected array structures with $ D_{1'} = 8 $, $ E_s^{(1)} = E_s^{(2)} = 30 $ dB and $ E_s^{(3)} = E_s^{(4)} = 20 $ dB}
\label{CDF_part}
\end{figure}

\section{Conclusion}

In this paper, we proposed an interference-aware two-stage beamforming framework, specifically for JSDM, where interference is considered in both analog and digital beamforming stages. Near-optimal slowly varying statistical analog beamformer (i.e., GEB) is designed where inter-group interference suppression is imposed at this stage. GEB is approximated with both fully and partially connected arrays by considering constant-modulus constraint. Furthermore, an algorithm to find the optimal connection structure is proposed for dynamic subarray design. Low-complexity alternating minimization algorithms are employed for obtaining constrained analog beamformers. Superiority of proposed beamformers is shown by using comprehensive analysis tools. It is observed that PE-AM attains the performance of unconstrained GEB with moderate interference strength. Furthermore, this algorithm outperforms commonly used DFT beamformer and phase of GEB. On the other hand, dynamic subarray design performs close to fully connected ones and outperforms fixed subarrays. Moreover, LMMSE type digital beamformers improve the performance of constrained analog beamformers as they take residual interference into account. This work could be extended to a spatial-wideband system as a future study.

\appendices
\vspace{-4mm}

\section{Covariance Matrices in Reduced Dimension}\label{cov_freq}

Covariance matrix of $ \tilde{\textbf{s}}_{f,k}^{(g)} $ is expressed as

\begin{equation}\label{E(lambda)}
\begin{aligned}
& \textbf{R}_{\tilde{\textbf{s}}_{f}}^{(g)} = \mathbb{E} \Big\{ \tilde{\textbf{s}}_{f,k}^{(g)} \big[ \tilde{\textbf{s}}_{f,k}^{(g)} \big]^H \Big\} \\
& = \mathbb{E}\Bigg\{ \frac{1}{\sqrt{N}} \sum_{n=0}^{N-1} \tilde{\textbf{s}}_{n}^{(g)} e^{-j \frac{2 \pi}{N} kn} \sum_{n'=0}^{N-1} \frac{1}{\sqrt{N}} \big[ \tilde{\textbf{s}}_{n'}^{(g)} \big]^H e^{j \frac{2 \pi}{N} kn'} \Bigg\} \\
& = \frac{1}{N} \sum_{n=0}^{N-1} \sum_{n'=0}^{N-1} \mathbb{E}\Big\{ \tilde{\textbf{s}}_{n}^{(g)} \big[ \tilde{\textbf{s}}_{n'}^{(g)} \big]^H \Big\} e^{j \frac{2 \pi}{N} k(n'-n)} \\
& = \frac{1}{N} \sum_{n=0}^{N-1} \sum_{n'=0}^{N-1} \big[ \textbf{S}^{(g)} \big]^H \textbf{R}_{\textbf{s}}^{(g)} \delta_{nn'} \textbf{S}^{(g)} e^{j \frac{2 \pi}{N} k(n'-n)} \\
& = \big[ \textbf{S}^{(g)} \big]^H \textbf{R}_{\textbf{s}}^{(g)} \textbf{S}^{(g)},
\end{aligned}
\end{equation}

\noindent
where $ \textbf{R}_{\textbf{s}}^{(g)} $ is the covariance matrix of intra-group signals of group $ g $ whose expression is given in \eqref{R_s}. Similarly, covariance matrix of $ \boldsymbol{\eta}_{f,k}^{(g)} $ can be expressed as $ \textbf{R}_{\tilde{\boldsymbol{\eta}}_{f}}^{(g)} = \big[ \textbf{S}^{(g)} \big]^H \textbf{R}_{\boldsymbol{\eta}}^{(g)} \textbf{S}^{(g)} $ where expression for $ \textbf{R}_{\boldsymbol{\eta}}^{(g)} $, which is the covariance matrix of sum of inter-group signals and noise, is given in \eqref{R_eta}. It is important to note that these covariance matrices are independent of frequency bin.

\section{Proof of Lemma \ref{LEM1}}\label{proof_lem1}

Replacing $ \textbf{S}^{(g)} $ with $ \textbf{S}^{(g)} \textbf{A} $ results in a similarity transformation for the second term in the determinant in \eqref{det}, i.e., $ \textbf{A}^{-1} \big( [ \textbf{S}^{(g)} ]^H \textbf{R}_{\boldsymbol{\eta}}^{(g)} \textbf{S}^{(g)} \big)^{-1} \big( [ \textbf{S}^{(g)} ]^H \textbf{R}_{\textbf{s}}^{(g)} \textbf{S}^{(g)} \big) \textbf{A} $. Similarity transformation does not affect the cost function since generalized eigenvectors diagonalize this term \cite{GMG}. Hence, alternative representations of GEB can be used without loss of performance.

\section{Expected SINR in Reduced Dimension}\label{stat_SINR}

In reduced dimension, it is possible to find an expected SINR value for each group. Intra-group signal power can be divided by power of sum of inter-group interference and noise in reduced dimension by using \eqref{y_n_g}. With this method, expected SINR of group $ g $ can be expressed as

\begin{equation}
\begin{aligned}
\overline{ SINR }^{(g)} & = \frac{ \mathbb{E} \bigg\{ \norm{ \tilde{\textbf{s}}_{f,k}^{(g)} }^2 \bigg\} }{ \mathbb{E} \bigg\{ \norm{ \tilde{\boldsymbol{\eta}}_{f,k}^{(g)} }^2 \bigg\} } = \frac{ \mathrm{Tr} \bigg\{ \mathbb{E} \Big\{ \tilde{\textbf{s}}_{f,k}^{(g)} \big[ \tilde{\textbf{s}}_{f,k}^{(g)} \big]^H  \Big\}  \bigg\}  }{ \mathrm{Tr} \bigg\{ \mathbb{E} \Big\{ \tilde{\boldsymbol{\eta}}_{f,k}^{(g)} \big[ \tilde{\boldsymbol{\eta}}_{f,k}^{(g)} \big]^H  \Big\}  \bigg\} } \\
& = \frac{ \mathrm{Tr} \Big\{ \textbf{R}_{ \tilde{\textbf{s}}_{f} }^{(g)} \Big\}  }{ \mathrm{Tr} \Big\{ \textbf{R}_{ \tilde{\boldsymbol{\eta}}_{f} }^{(g)}  \Big\} } = \frac{ \mathrm{Tr} \Big\{ \big[ \textbf{S}^{(g)} \big]^H \textbf{R}_{\textbf{s}}^{(g)} \textbf{S}^{(g)} \Big\} }{ \mathrm{Tr} \Big\{ \big[ \textbf{S}^{(g)} \big]^H \textbf{R}_{\eta}^{(g)} \textbf{S}^{(g)}  \Big\} }, 
\end{aligned}\label{SINR_bar}
\end{equation}

\noindent
where expressions for covariance matrices $ \textbf{R}_{ \tilde{\textbf{s}}_{f} }^{(g)} $ and $ \textbf{R}_{ \tilde{\boldsymbol{\eta}}_{f} }^{(g)} $ are taken from Appendix \ref{cov_freq}.

\section{Channel Estimators}\label{channel}

\subsection{LMMSE Type Channel Estimator}

LMMSE type channel estimator of group $ g $ can be written as

\begin{equation}
\textbf{Z}^{(g)} = \textbf{R}_{\bar{\textbf{y}}^{(g)}}^{-1} \textbf{R}_{\bar{\textbf{y}}^{(g)} \bar{\textbf{h}}_{eff}^{(g,g)}} \label{LMMSE}
\end{equation}

\noindent
where covariance matrix between $ \bar{\textbf{y}}^{(g)} $ and $ \bar{\textbf{h}}_{eff}^{(g,g)} $ and covariance matrix of $ \bar{\textbf{y}}^{(g)} $ can be computed as

\begin{equation}
\textbf{R}_{\bar{\textbf{y}}^{(g)} \bar{\textbf{h}}_{eff}^{(g,g)}} = \mathbb{E} \Big\{ \bar{\textbf{y}}^{(g)} \big[ \bar{\textbf{h}}_{eff}^{(g,g)} \big]^H \Big\} = \Big( \textbf{X}^{(g)} \otimes \textbf{I}_{D_{g}} \Big) \textbf{R}_{ \bar{\textbf{h}}_{eff}^{(g,g)}}, \label{R_y_g_h_eff}
\end{equation}

\begin{equation}
\begin{aligned}
& \! \! \textbf{R}_{\bar{\textbf{y}}^{(g)}} = \mathbb{E} \Big\{ \bar{\textbf{y}}^{(g)} \big[ \bar{\textbf{y}}^{(g)} \big]^H \Big\} \\ & \: \: = \Big( \textbf{X}^{(g)} \otimes \textbf{I}_{D_{g}} \Big) \textbf{R}_{ \bar{\textbf{h}}_{eff}^{(g,g)}} \Big( \textbf{X}^{(g)} \otimes \textbf{I}_{D_{g}} \Big)^H
+ \textbf{I}_{T} \otimes \textbf{R}_{\tilde{\textbf{\eta}}}^{(g)},
\end{aligned} \label{R_y_g_bar}
\end{equation}

\noindent
where covariance matrix of sum of inter-group interference and noise terms in reduced dimension is denoted by $ \mathbb{E} \Big\{ \tilde{\textbf{\eta}}_{n}^{(g)} \big[ \tilde{\textbf{\eta}}_{n'}^{(g)} \big]^H \Big\} = \textbf{R}_{\tilde{\textbf{\eta}}}^{(g)} \delta_{nn'} $. This covariance matrix can be expressed as $ \textbf{R}_{\tilde{\textbf{\eta}}}^{(g)} = \big[ \textbf{S}^{(g)} \big]^H \textbf{R}_{\textbf{\eta}}^{(g)} \textbf{S}^{(g)} $. Other groups are assumed to be in data transmission mode in \eqref{R_y_g_bar} according to Remark \ref{rem3}. That is, $ \textbf{X}^{(g')} $ is taken as a random matrix for $ g' \neq g $ which is why $ \textbf{R}_{\tilde{\textbf{\eta}}}^{(g)} $ is used in \eqref{R_y_g_bar}. Covariance matrix of the overall effective channel of group $ g $ is expressed as

\begin{equation}
\begin{aligned}
& \! \! \! \! \textbf{R}_{ \bar{\textbf{h}}_{eff}^{(g,g)}} = \mathbb{E} \Big\{ \bar{\textbf{h}}_{eff}^{(g,g)} \big[ \bar{\textbf{h}}_{eff}^{(g,g)} \big]^H \Big\} \\ & \: \: \: \: = \mathrm{blkdiag} \Bigg\{ \sum_{l=0}^{L-1} \textbf{E}_{L,l+1} \otimes \big[ \textbf{S}^{(g)} \big]^H \textbf{R}_{l}^{(g_m)} \textbf{S}^{(g)} \Bigg\}_{m=1}^{K_g} 
\end{aligned} \label{R_h_eff}
\end{equation}

\noindent
where $ \textbf{E}_{L,l+1} \triangleq \textbf{e}_{L,l+1} \textbf{e}_{L,l+1}^H $ and $ \textbf{e}_{L,l+1} \triangleq \big[ \textbf{I}_L \big]_{(:,l+1)} $. It is important to note that due to uncorrelated structure of channel vectors, \eqref{R_h_eff} has a block diagonal form where blocks in the diagonals are covariance matrices of effective channels.

\subsection{LS Type Channel Estimator}

LS type channel estimator of group $ g $ can be expressed as

\begin{equation}
\textbf{Z}^{(g)} = \bigg( \textbf{X}^{(g)} \Big( \big[ \textbf{X}^{(g)} \big]^H \textbf{X}^{(g)} \Big)^{-1} \bigg) \otimes \textbf{I}_{D_g}. \label{LS}
\end{equation}

However, $ \bar{\textbf{h}}_{eff}^{(g,g)} $ has a sparse structure which degrades the performance of LS estimator given in \eqref{LS}. We should eliminate the columns of $ \textbf{X}^{(g)} $ corresponding to the inactive MPCs and find the estimator in \eqref{LS}. Then, we should insert zero columns to the estimator to make up for the eliminated columns. These zero columns in $ \textbf{Z}^{(g)} $ leads to estimating zero vectors for channels belonging to inactive MPCs. In this way, the dimension of the inverse in the LS estimator is reduced and performance of the estimator is increased.

\bibliographystyle{IEEEtran}
\bibliography{References}

\begin{thebibliography}{10}
\providecommand{\url}[1]{#1}
\csname url@samestyle\endcsname
\providecommand{\newblock}{\relax}
\providecommand{\bibinfo}[2]{#2}
\providecommand{\BIBentrySTDinterwordspacing}{\spaceskip=0pt\relax}
\providecommand{\BIBentryALTinterwordstretchfactor}{4}
\providecommand{\BIBentryALTinterwordspacing}{\spaceskip=\fontdimen2\font plus
\BIBentryALTinterwordstretchfactor\fontdimen3\font minus
  \fontdimen4\font\relax}
\providecommand{\BIBforeignlanguage}[2]{{%
\expandafter\ifx\csname l@#1\endcsname\relax
\typeout{** WARNING: IEEEtran.bst: No hyphenation pattern has been}%
\typeout{** loaded for the language `#1'. Using the pattern for}%
\typeout{** the default language instead.}%
\else
\language=\csname l@#1\endcsname
\fi
#2}}
\providecommand{\BIBdecl}{\relax}
\BIBdecl

\bibitem{mMIMO}
E.~G. {Larsson}, O.~{Edfors}, F.~{Tufvesson}, and T.~L. {Marzetta}, ``Massive
  {MIMO} for next generation wireless systems,'' \emph{IEEE Commun. Mag.},
  vol.~52, no.~2, pp. 186--195, 2014.

\bibitem{mMIMOoverview}
L.~Lu, G.~Y. Li, A.~L. Swindlehurst, A.~Ashikhmin, and R.~Zhang, ``An overview
  of massive {MIMO}: Benefits and challenges,'' \emph{IEEE J. Sel. Topics
  Signal Process.}, vol.~8, no.~5, pp. 742--758, 2014.

\bibitem{SPmmWave}
R.~W. Heath, N.~Gonzalez-Prelcic, S.~Rangan, W.~Roh, and A.~M. Sayeed, ``An
  overview of signal processing techniques for millimeter wave {MIMO}
  systems,'' \emph{IEEE J. Sel. Topics Signal Process.}, vol.~10, no.~3, pp.
  436--453, 2016.

\bibitem{SpaSparse}
O.~E. {Ayach}, S.~{Rajagopal}, S.~{Abu-Surra}, Z.~{Pi}, and R.~W. {Heath},
  ``Spatially sparse precoding in millimeter wave {MIMO} systems,'' \emph{IEEE
  Trans. Wireless Commun.}, vol.~13, no.~3, pp. 1499--1513, 2014.

\bibitem{LowComplexityHybrid}
C.~{Rusu}, R.~{Mèndez-Rial}, N.~{González-Prelcic}, and R.~W. {Heath}, ``Low
  complexity hybrid precoding strategies for millimeter wave communication
  systems,'' \emph{IEEE Trans. Wireless Commun.}, vol.~15, no.~12, pp.
  8380--8393, 2016.

\bibitem{AltMin}
X.~{Yu}, J.~{Shen}, J.~{Zhang}, and K.~B. {Letaief}, ``Alternating minimization
  algorithms for hybrid precoding in millimeter wave {MIMO} systems,''
  \emph{IEEE J. Sel. Topics Signal Process.}, vol.~10, no.~3, pp. 485--500,
  2016.

\bibitem{WeiYu}
F.~{Sohrabi} and W.~{Yu}, ``Hybrid digital and analog beamforming design for
  large-scale antenna arrays,'' \emph{IEEE J. Sel. Topics Signal Process.},
  vol.~10, no.~3, pp. 501--513, 2016.

\bibitem{HybridSurvey}
I.~{Ahmed}, H.~{Khammari}, A.~{Shahid}, A.~{Musa}, K.~S. {Kim}, E.~{De
  Poorter}, and I.~{Moerman}, ``A survey on hybrid beamforming techniques in
  {5G}: Architecture and system model perspectives,'' \emph{IEEE Commun.
  Surveys Tuts.}, vol.~20, no.~4, pp. 3060--3097, 2018.

\bibitem{JSDM}
A.~{Adhikary}, J.~{Nam}, J.~{Ahn}, and G.~{Caire}, ``Joint spatial division and
  multiplexing—{The} large-scale array regime,'' \emph{IEEE Trans. Inf.
  Theory}, vol.~59, no.~10, pp. 6441--6463, 2013.

\bibitem{JSDMmmWave}
A.~{Adhikary}, E.~{Al Safadi}, M.~K. {Samimi}, R.~{Wang}, G.~{Caire}, T.~S.
  {Rappaport}, and A.~F. {Molisch}, ``Joint spatial division and multiplexing
  for mm-{Wave} channels,'' \emph{IEEE J. Sel. Areas Commun.}, vol.~32, no.~6,
  pp. 1239--1255, June 2014.

\bibitem{MessagePassing}
N.~J. {Myers} and R.~W. {Heath}, ``Message passing-based joint {CFO} and
  channel estimation in {mmWave} systems with one-bit {ADCs},'' \emph{IEEE
  Trans. Wireless Commun.}, vol.~18, no.~6, pp. 3064--3077, 2019.

\bibitem{MUD}
S.~Wang, Y.~Li, and J.~Wang, ``Multiuser detection in massive {MIMO} with
  quantized phase-only measurements,'' in \emph{Proc. IEEE Int. Conf.
  Commun.}\hskip 1em plus 0.5em minus 0.4em\relax IEEE, 2015, pp. 4576--4581.

\bibitem{PAPR}
H.~G. {Myung}, J.~{Lim}, and D.~J. {Goodman}, ``Peak-to-average power ratio of
  single carrier {FDMA} signals with pulse shaping,'' in \emph{Proc. IEEE 17th
  Int. Symp. Pers., Indoor and Mobile Radio Commun.}, 2006, pp. 1--5.

\bibitem{SC-FDE}
F.~{Pancaldi}, G.~M. {Vitetta}, R.~{Kalbasi}, N.~{Al-Dhahir}, M.~{Uysal}, and
  H.~{Mheidat}, ``Single-carrier frequency domain equalization,'' \emph{IEEE
  Signal Process. Mag.}, vol.~25, no.~5, pp. 37--56, 2008.

\bibitem{SCvsOFDM}
S.~Buzzi, C.~D’Andrea, T.~Foggi, A.~Ugolini, and G.~Colavolpe,
  ``Single-carrier modulation versus {OFDM} for millimeter-wave wireless
  {MIMO},'' \emph{IEEE Trans. Commun.}, vol.~66, no.~3, pp. 1335--1348, 2017.

\bibitem{SCoptimality}
A.~Pitarokoilis, S.~K. Mohammed, and E.~G. Larsson, ``On the optimality of
  single-carrier transmission in large-scale antenna systems,'' \emph{IEEE
  Wireless Commun. Lett.}, vol.~1, no.~4, pp. 276--279, 2012.

\bibitem{mmWaveSC_Hybrid_mMIMO}
X.~Song, S.~Haghighatshoar, and G.~Caire, ``Efficient beam alignment for
  millimeter wave single-carrier systems with hybrid {MIMO} transceivers,''
  \emph{IEEE Trans. Wireless Commun.}, vol.~18, no.~3, pp. 1518--1533, 2019.

\bibitem{GC}
G.~M. {Guvensen} and E.~{Ayanoglu}, ``A generalized framework on beamformer
  design and {CSI} acquisition for single-carrier massive {MIMO} systems in
  millimeter wave channels,'' in \emph{Proc. IEEE Globecom Workshops}, 2016,
  pp. 1--7.

\bibitem{Anil}
A.~{Kurt} and G.~M. {Guvensen}, ``An efficient hybrid beamforming and channel
  acquisition for wideband mm-{Wave} massive {MIMO} channels,'' in \emph{Proc.
  IEEE Int. Conf. Commun.}, 2019, pp. 1--7.

\bibitem{NewJSDM}
Y.~{Jeon}, C.~{Song}, S.~{Lee}, S.~{Maeng}, J.~{Jung}, and I.~{Lee}, ``New
  beamforming designs for joint spatial division and multiplexing in
  large-scale {MISO} multi-user systems,'' \emph{IEEE Trans. Wireless Commun.},
  vol.~16, no.~5, pp. 3029--3041, 2017.

\bibitem{DFT}
A.~{Liu} and V.~{Lau}, ``Phase only {RF} precoding for massive {MIMO} systems
  with limited {RF} chains,'' \emph{IEEE Trans. Signal Process.}, vol.~62,
  no.~17, pp. 4505--4515, 2014.

\bibitem{TwoStage}
J.~{Choi}, G.~{Lee}, and B.~L. {Evans}, ``Two-stage analog combining in hybrid
  beamforming systems with low-resolution {ADCs},'' \emph{IEEE Trans. Signal
  Process.}, vol.~67, no.~9, pp. 2410--2425, 2019.

\bibitem{Exploiting}
S.~{Park}, J.~{Park}, A.~{Yazdan}, and R.~W. {Heath}, ``Exploiting spatial
  channel covariance for hybrid precoding in massive {MIMO} systems,''
  \emph{IEEE Trans. Signal Process.}, vol.~65, no.~14, pp. 3818--3832, 2017.

\bibitem{Dynamic}
S.~{Park}, A.~{Alkhateeb}, and R.~W. {Heath}, ``Dynamic subarrays for hybrid
  precoding in wideband {mmWave} {MIMO} systems,'' \emph{IEEE Trans. Wireless
  Commun.}, vol.~16, no.~5, pp. 2907--2920, 2017.

\bibitem{FixedPart}
S.~{He}, C.~{Qi}, Y.~{Wu}, and Y.~{Huang}, ``Energy-efficient transceiver
  design for hybrid sub-array architecture {MIMO} systems,'' \emph{IEEE
  Access}, vol.~4, pp. 9895--9905, 2016.

\bibitem{ChStat}
J.~{Jin}, C.~{Xiao}, W.~{Chen}, and Y.~{Wu}, ``Channel-statistics-based hybrid
  precoding for millimeter-wave {MIMO} systems with dynamic subarrays,''
  \emph{IEEE Trans. Commun.}, vol.~67, no.~6, pp. 3991--4003, 2019.

\bibitem{DynamicWideband}
F.~{Yang}, J.~B. {Wang}, M.~{Cheng}, J.~Y. {Wang}, M.~{Lin}, and J.~{Cheng},
  ``A partially dynamic subarrays structure for wideband {mmWave} {MIMO}
  systems,'' \emph{IEEE Trans. Commun.}, vol.~68, no.~12, pp. 7578--7592, 2020.

\bibitem{DynamicWideband2}
H.~{Li}, M.~{Li}, Q.~{Liu}, and A.~L. {Swindlehurst}, ``Dynamic hybrid
  beamforming with low-resolution {PSs} for wideband {mmWave} {MIMO}-{OFDM}
  systems,'' \emph{IEEE J. Sel. Areas Commun.}, vol.~38, no.~9, pp. 2168--2181,
  2020.

\bibitem{DynamicMISO}
H.~{Li}, M.~{Li}, and Q.~{Liu}, ``Hybrid beamforming with dynamic subarrays and
  low-resolution {PSs} for {mmWave} {MU}-{MISO} systems,'' \emph{IEEE Trans.
  Commun.}, vol.~68, no.~1, pp. 602--614, 2020.

\bibitem{Switch1}
R.~{Méndez-Rial}, C.~{Rusu}, N.~{González-Prelcic}, A.~{Alkhateeb}, and R.~W.
  {Heath}, ``Hybrid {MIMO} architectures for millimeter wave communications:
  Phase shifters or switches?'' \emph{IEEE Access}, vol.~4, pp. 247--267, 2016.

\bibitem{Switch2}
S.~{Buzzi}, C.~{I}, T.~E. {Klein}, H.~V. {Poor}, C.~{Yang}, and A.~{Zappone},
  ``A survey of energy-efficient techniques for {5G} networks and challenges
  ahead,'' \emph{IEEE J. Sel. Areas Commun.}, vol.~34, no.~4, pp. 697--709,
  2016.

\bibitem{JSDM_Group}
J.~{Nam}, A.~{Adhikary}, J.~{Ahn}, and G.~{Caire}, ``Joint spatial division and
  multiplexing: Opportunistic beamforming, user grouping and simplified
  downlink scheduling,'' \emph{IEEE J. Sel. Topics Signal Process.}, vol.~8,
  no.~5, pp. 876--890, 2014.

\bibitem{Joint_Group}
J.~{Chen} and D.~{Gesbert}, ``Joint user grouping and beamforming for low
  complexity massive {MIMO} systems,'' in \emph{Proc. IEEE 17th Int. Workshop
  on Signal Process. Advances in Wireless Commun. (SPAWC)}, 2016, pp. 1--6.

\bibitem{CCM_Caire1}
S.~{Haghighatshoar}, M.~B. {Khalilsarai}, and G.~{Caire}, ``Multi-band
  covariance interpolation with applications in massive {MIMO},'' in
  \emph{Proc. IEEE Int. Symp. Inf. Theory (ISIT)}, 2018, pp. 386--390.

\bibitem{CCM_Gao}
H.~{Xie}, F.~{Gao}, S.~{Jin}, J.~{Fang}, and Y.~C. {Liang}, ``Channel
  estimation for {TDD/FDD} massive {MIMO} systems with channel covariance
  computing,'' \emph{IEEE Trans. Wireless Commun.}, vol.~17, no.~6, pp.
  4206--4218, 2018.

\bibitem{CCM_Caire2}
M.~{Barzegar Khalilsarai}, S.~{Haghighatshoar}, X.~{Yi}, and G.~{Caire},
  ``{FDD} massive {MIMO} via {UL/DL} channel covariance extrapolation and
  active channel sparsification,'' \emph{IEEE Trans. Wireless Commun.},
  vol.~18, no.~1, pp. 121--135, 2019.

\bibitem{CCM_GMG}
A.~O. {Kalayci} and G.~M. {Guvensen}, ``An efficient spatial channel covariance
  estimation via joint angle-delay power profile in hybrid massive {MIMO}
  systems,'' in \emph{Proc. IEEE Int. Conf. Commun. Workshops}, 2020, pp. 1--7.

\bibitem{SpatialWideband}
B.~{Wang}, F.~{Gao}, S.~{Jin}, H.~{Lin}, and G.~Y. {Li}, ``Spatial- and
  frequency-wideband effects in millimeter-wave massive {MIMO} systems,''
  \emph{IEEE Trans. Signal Process.}, vol.~66, no.~13, pp. 3393--3406, 2018.

\bibitem{HybridBeamSquint}
Y.~{Chen}, Y.~{Xiong}, D.~{Chen}, T.~{Jiang}, S.~X. {Ng}, and L.~{Hanzo},
  ``Hybrid precoding for wideband millimeter wave {MIMO} systems in the face of
  beam squint,'' \emph{IEEE Trans. Wireless Commun.}, vol.~20, no.~3, pp.
  1847--1860, 2021.

\bibitem{BeamSquintComp}
Y.~{Chen}, D.~{Chen}, T.~{Jiang}, and L.~{Hanzo}, ``Channel-covariance and
  angle-of-departure aided hybrid precoding for wideband multiuser millimeter
  wave {MIMO} systems,'' \emph{IEEE Trans. Commun.}, vol.~67, no.~12, pp.
  8315--8328, 2019.

\bibitem{GMG}
G.~M. {Guvensen}, C.~{Candan}, S.~{Koc}, and U.~{Orguner}, ``On generalized
  eigenvector space for target detection in reduced dimensions,'' in
  \emph{Proc. IEEE Radar Conf.}, 2015, pp. 1316--1321.

\bibitem{OPP}
P.~H. Sch{\"o}nemann, ``A generalized solution of the orthogonal procrustes
  problem,'' \emph{Psychometrika}, vol.~31, no.~1, pp. 1--10, 1966.

\bibitem{Bussgang}
O.~T. {Demir} and E.~{Bjornson}, ``The bussgang decomposition of nonlinear
  systems: Basic theory and {MIMO} extensions [lecture notes],'' \emph{IEEE
  Signal Process. Mag.}, vol.~38, no.~1, pp. 131--136, 2021.

\end{thebibliography}

\begin{IEEEbiography}[{\includegraphics[width=1in,height=1.25in,clip,keepaspectratio]{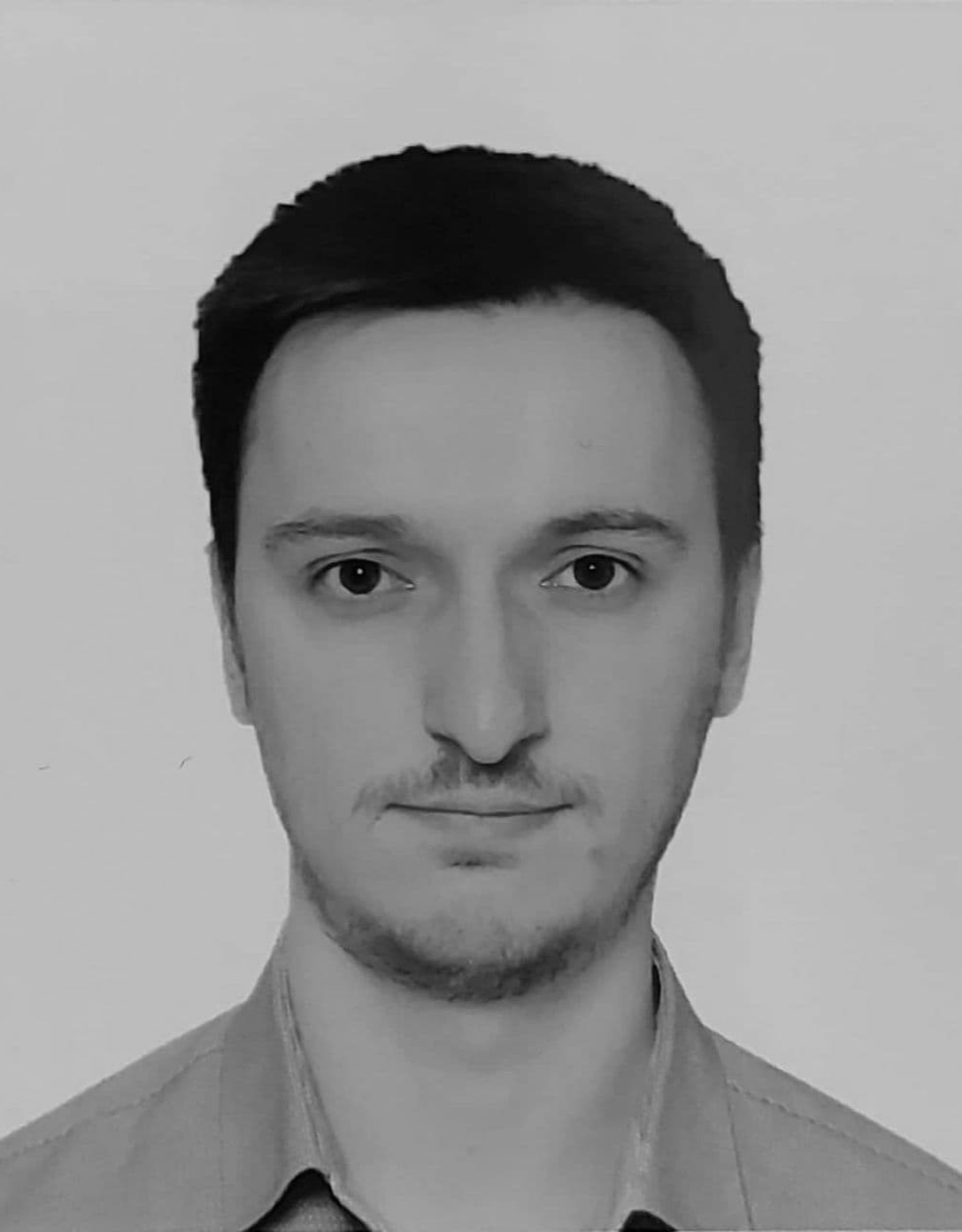}}]{Murat Bayraktar} (S'20) received the B.S. degrees in
electrical and electronics engineering and physics from the Middle East Technical University (METU), Ankara,
Turkey, in 2018. He is currently working toward the
M.S. degree at the Department of Electrical and Electronics Engineering, METU. Since 2019, he has been a research and teaching assistant with the same department. His current research interests include signal processing for wireless communications with particular focus on mm-wave massive MIMO techniques, hybrid analog/digital systems and non-orthogonal multiple access (NOMA) methods.
\end{IEEEbiography}

\begin{IEEEbiography}[{\includegraphics[width=1in,height=1.25in,clip,keepaspectratio]{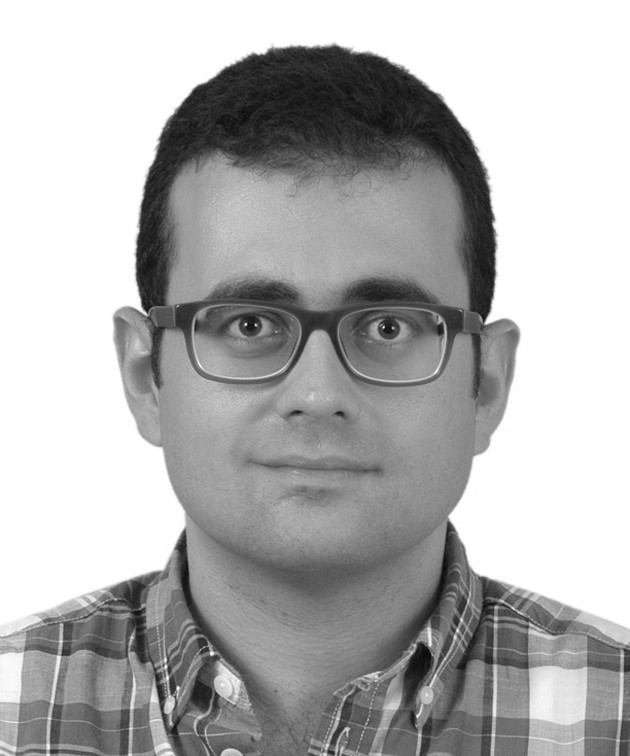}}]{Gokhan M. Guvensen} received his B.S., M.S., and Ph.D. degrees in electrical and electronics engineering from the Middle East Technical University (METU), Ankara, Turkey in 2006, 2009 and 2014, respectively. He worked as a postdoctoral fellow in the Center for Pervasive Communications and Computing (CPCC) in the University of California, Irvine (UCI), USA between 2015 and 2016. In 2017, he joined the Electrical and Electronics Engineering Department at METU, where he is now an Assistant Professor. His research interests include the design of digital communication systems and statistical signal processing with a particular focus on modulation theory, next-generation mobile communication techniques, iterative detection and equalization techniques, information theory, and radar signal processing.
 
\end{IEEEbiography}

\EOD

\end{document}